\pdfoutput=1
\RequirePackage{etoolbox}
\documentclass[acmtog,screen,table,authorversion]{acmart}

\usepackage{booktabs} % For formal tables
\usepackage{cancel}

\usepackage{geometrycollective}
\usepackage{enumitem}
\usepackage{kbordermatrix}
\usepackage{multirow}

\usepackage{xcolor,colortbl}
\definecolor{gOrange}{HTML}{ce8e00}
\definecolor{gBlue}{HTML}{0076b9}
\definecolor{gGreen}{HTML}{008c61}
\definecolor{accent}{HTML}{168587}
\definecolor{aliceblue}{rgb}{0.94, 0.97, 1.0}

\newcommand{\connectionUnitary}{r}
\newcommand{\connectionAngle}{\rho}
\newcommand{\fieldRotationAngle}{\omega}

\newcommand{\DirichletEnergy}{\mathcal{E}^D}
\newcommand{\CircularWellPotential}{\mathcal{W}}
\newcommand{\DiscreteGinzburgLandau}{\mathcal{G}\hspace{-.5mm}\mathcal{L}}
\newcommand{\Landmarks}{\mathcal{L}}

\DeclareMathOperator*{\argmin}{arg\,min}

\setcopyright{cc}
\setcctype{by}
\acmJournal{TOG}
\acmYear{2026} \acmVolume{45} \acmNumber{4} \acmArticle{162}
\acmMonth{7} \acmDOI{10.1145/3811368}

\setcopyright{acmcopyright}
\acmSubmissionID{1122}
\graphicspath{{./images_arxiv/}}

\begin{document}
% Title portion
\title{Implicit Minimal Surfaces for Bijective Correspondences}
\author{Etienne Corman}
\email{etienne.corman@cnrs.fr}
\orcid{0009-0002-9401-2362}
\affiliation{%
 \institution{Universit\'{e} de Lorraine, CNRS, Inria, LORIA}
 \streetaddress{615 Rue du Jardin-Botanique}
 \city{Vand\oe uvre-l\`{e}s-Nancy}
 \country{France}
 \postcode{54506}
}
\author{Yousuf Soliman}
\email{yousufs@sidefx.com}
\orcid{0000-0003-4023-5026}
\affiliation{%
   \institution{Side Effects Software}
   \city{Toronto}
   \country{Canada}
}
\author{Robin Magnet}
\email{robin.magnet@inria.fr}
\orcid{0000-0002-2192-411X}
\affiliation{%
  \institution{Inria, Universit\'{e} Paris Cit\'{e}}
  \city{Paris}
  \country{France}
}
\author{Mark Gillespie}
\email{mark.gillespie81@gmail.com}
\orcid{0009-0000-5645-9636}
\affiliation{%
  \institution{Inria}
  \city{Palaiseau}
  \country{France}
}
\affiliation{%
  \institution{University of Utah}
  \city{Salt Lake City}
  \country{USA}
}

\renewcommand\shortauthors{Corman et al.}

\begin{abstract}
We introduce an \emph{implicit} representation of continuous, bijective, orientation-preserving maps between genus zero surfaces with or without boundary. The distortion of these maps can easily be minimized by optimizing the Ginzburg-Landau functional---a ubiquitous model in physics and differential geometry---leading to a simple algorithm for computing bijective correspondences using only standard tools of the tangent vector field toolbox. The method avoids combinatorial mesh modifications and does not require barrier functions to enforce bijectivity making it more robust to noise and simpler to implement. Moreover, the algorithm does not assume a bijective initialization and can untangle non-bijective correspondences generated by computationally cheaper methods such as functional maps. It supports the use of both landmark points and landmark curves to guide the correspondence.
The key idea is that a bijection between surfaces defines a two-dimensional mapping surface sitting inside the four-dimensional product space of the two inputs, and this mapping surface can be stored implicitly as the zero set of a \emph{complex section}---essentially a complex function defined on the product space. Now the distortion of the map can be optimized by minimizing the area of this mapping surface, which amounts to minimizing the Ginzburg-Landau functional of the complex section. We demonstrate the practical benefits of our method by comparing to state-of-the-art correspondence algorithms and show that our implicit representation offers improved stability and naturally supports constraints that are difficult to enforce with explicit map representations.
\end{abstract}




\maketitle

% decrease the amount of spacing above fig
\setlength{\intextsep}{0.75em}
% set wrapfig appearance here so it doesn't mess with paper margins
\makeatletter
\patchcmd\WF@putfigmaybe{\lower\intextsep}{}{}{\fail}%
\AddToHook{env/wrapfigure/begin}{\setlength{\intextsep}{0pt}} % no space above wrapfig
\makeatother
\setlength{\columnsep}{0.5em} % margin size between text and wrapfig

\begin{figure}
    \centering
    \includegraphics{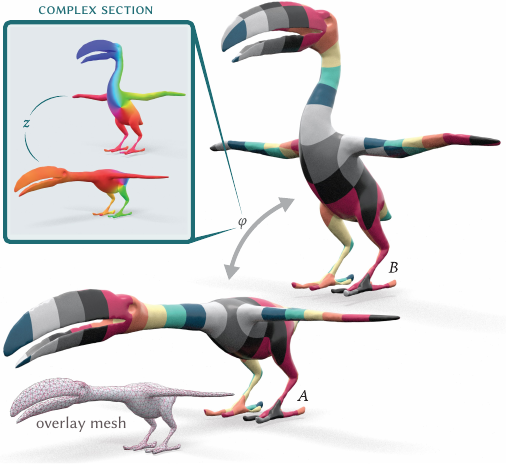}
    \caption{We introduce a new implicit representation of maps \(\varphi : A \to B\) between triangle meshes \(A\) and \(B\), encoding the maps as the zero sets of complex functions \(z\) on the product space \(A \times B\). In this representation we can compute high-quality orientation-preserving bijections between \(A\) and \(B\) by minimizing a simple Ginzburg-Landau energy, without requiring any combinatorial mesh modifications, barrier functions, or a bijective initialization. Once an implicit map has been computed, it encodes not only the vertex map, but also the entire overlay mesh under the correspondence.}
    \label{fig:Teaser}
\end{figure}

\section{Introduction}
\label{sec:Introduction}

% Mapping between surfaces are fundamental
Computing a map between two surfaces \(A\) and \(B\) is a ubiquitous problem in geometry processing. It is a prerequisite for analyzing data across collections of shapes and is essential for transferring data such as textures, segmentations, or other semantic attributes, which are crucial to a wide range of applications from studying brain-fold variations in medical imaging to reconstructing morphological correspondences in paleontology, and many more tasks across computer graphics, computer vision, and scientific modeling.

% For "clean" data we want bijective continuous distortion minimizing maps
When the surfaces share the same topological class, the goal is typically to find a map that is bijective, continuous, and aligns salient geometric or semantic features---often formalized by looking for as-isometric-as possible maps. Yet it remains highly challenging to compute such maps: the underlying distortion energies are non-convex, while continuity and bijectivity are strict constraints that are difficult to enforce during optimization.

% Describing bijective and continuous maps is hard
In practice, even describing an exact map between two triangle meshes is cumbersome, yet it is necessary for evaluating a map's distortion. For every point on $A$, the map must specify its image on $B$, not only for the vertices of $A$ but also for all interior points of its faces. The most common approach is to describe the map by \emph{remeshing} both surfaces so that they share identical connectivity, either by intersecting the meshes~\cite{schmidt2020inter} or by replacing them altogether~\cite{schmidt2023surface}. However, this representation is deeply coupled to the mesh connectivity, whereas the map itself is a smooth, coordinate-invariant geometric object. Instead, we propose a new representation of correspondences that avoids combinatorial modifications entirely and relies only on well-established tools for processing tangent vector fields (\cref{fig:Teaser}).

\begin{figure}
    \includegraphics{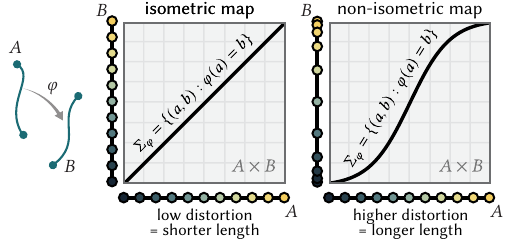}
    \caption{To illustrate the idea behind our method, we move down a dimension and consider a map \(\varphi : A \to B\) between two curves \figloc{(left)}. Such a map can be represented by its ``graph'', \ie{} the set of identified points \({\Sigma_\varphi = \{(a, b)\;:\; \varphi(a) = b\}}\) inside the product space \(\subset A \times B\) \figloc{(center, right)}. The more that \(\varphi\) distorts the curves, the longer its graph becomes---thus, we can compute high-quality maps by minimizing the length of the graph.}
    \label{fig:minimal_surface_2d}
\end{figure}

% We propose to go to 4D: bijective continuous correspondences are a continuous surface in the product space
\paragraph{Mappings as Surfaces in Product Space}
Any map $\varphi : A \rightarrow B$ carves out a two-dimensional manifold $\Sigma_{\varphi}$ from the product space $A \times B$:
\begin{equation}
    \Sigma_{\varphi} \coloneqq \{ (p, \varphi(p)) \;:\; p \in A \}.
    \label{eq:GraphDefinition}
\end{equation}
When $\varphi$ is continuous and bijective each point $p \in A$ is mapped to a unique point $\varphi(p) \in B$, and so the projection of $\Sigma_{\varphi}$ onto $B$ covers $B$ exactly once. Conversely, if a two-dimensional submanifold $\Sigma \subset A \times B$ has projections onto each input surface which cover that surface exactly once, then $\Sigma$ represents a continuous bijection between $A$ and $B$. Moreover, if the projections are orientation-preserving, then the bijection will also preserve orientation (\cref{fig:SymmetryMap}). And if $\Sigma$ has minimal area, it represents a distortion-minimizing map (\cref{fig:minimal_surface_4d}). (For the mathematical expression of this distortion, and its relation to other well known distortion measures, see \cref{sec:AreaMinimizationAndMetricDistortion}.)

We can illustrate this idea for the simpler problem of matching curves (\cref{fig:minimal_surface_2d}). In this setting, the product space is a rectangle, and a map is represented by the graph of the function $a \mapsto \varphi(a)$. This representation enables a direct measure of the distortion of the map: the length of the graph. The identity map, represented by a straight line, minimizes distortion, while any other mapping necessarily increases both the distortion and the length of the graph.

We thus shift the focus: rather than working with map $\varphi$ directly, we instead compute a surface in the product manifold $A \times B$ which implicitly encodes a distortion-minimizing correspondence.

\paragraph{Minimal Surfaces and Complex Fields}
% Minimal surfaces in 4D are computed by minimizing the GL functional (known results, intuition)
Now that we have framed the matching problem using minimal surfaces, a central question remains:  \emph{How can we effectively represent and compute a minimal surface in a four-dimensional manifold?} Fortunately a substantial body of literature addresses this problem using tools remarkably similar to those used in geometry processing. A key insight is that minimal surfaces can be characterized through the zero set of a vector-field like object known as a \emph{complex section}. In geometry processing, we often use the property that the zeros of a smooth vector field on a surface lie at isolated points, and the zeros of complex functions in 3D space lie along curves. Similarly, the zeros of a complex section on a 4D space lie on a two-dimensional subset.

This leads to our core representation: a map \(\varphi : A \to B\) is encoded as a complex-valued section $z$ on the product space \(A \times B\) which is zero on the mapping surface: so \(z(a, b) = 0\) if and only if \(b = \varphi(a)\).

The bijectivity of \(\varphi\) can also be expressed via the field \(z\). For each fixed $p \in A$, the restriction $b \mapsto z(p, b)$ must vanish at exactly one point $b \in B$, and vice versa for each fixed \(q \in B\). This condition is incorporated into a special choice of connection on $A$ and $B$, akin to the standard treatment of cross fields (see~\cref{sec:ImplicitRepresentationViaComplexLineBundles}). Strictly speaking, this connection does not force the output map to be bijective---it only imposes the relaxed constraint that the net number of signed zeros is equal to one for each restriction. But we find that in practice this relaxed condition suffices to compute high-quality maps (\cref{sec:Results}).

Using this representation, the computation of a minimal surface becomes accessible. Minimizers of the Ginzburg-Landau functional
\begin{equation}
    \DiscreteGinzburgLandau_{\varepsilon}(z) := \int_{A \times B}\big(\tfrac{1}{2}|\nabla z|^2 + \frac{1}{4\varepsilon^2}(1-|z|^2)^2\big)\vol_{A \times B},
\end{equation}
converge, as $\varepsilon \to 0$, to fields whose zero sets form minimal surfaces. Intuitively, if we were to normalize the field \(z\), zeros of \(z\) would become singularities where the Dirichlet energy blows up. The parameter $\varepsilon$ forces the field to approach unit length almost everywhere, leading to a large Dirichlet energy near the zeros. Thus, minimizing Dirichlet energy minimizes the area of the zero set. A more detailed discussion of this phenomenon is provided in~\cref{sec:ImplicitAreaMinimization}.

The behavior of the Ginzburg-Landau functional and related energies on spaces with nontrivial topology is still an active area of mathematical research, so in this work we mostly restrict our attention to surfaces with \emph{sphere-like topology}. However, we discuss the treatment of genus-zero surfaces with boundary in \cref{sec:SurfacesWithBoundary}.

\paragraph{Contributions}
In summary, we introduce a method to compute high quality continuous, bijective, orientation-preserving maps between genus zero surfaces with or without boundary, by minimizing the Ginzburg-Landau energy on a four-dimensional product space.
%Although this objective is non-convex, it offers a key advantage: it does not require any of the barrier terms commonly used to enforce injectivity~\cite{schuller2013injective}. Consequently, unlike many existing approaches, our method does not require a bijective initialization and can naturally ``untangle'' correspondences produced by non-bijective or noisy inputs. In addition, the formulation is fully intrinsic while inherently preserving orientation.
%
This approach yields several advantages:
\begin{itemize}[leftmargin={5mm}]
    \item Distortion minimization is achieved without barrier terms that prohibit non-bijective configurations, resulting in improved stability and simpler optimization;
    \item The method does not require a bijective initial map and can robustly ``untangle'' invalid or overlapping correspondences;
    \item In addition to landmark points, the algorithm naturally supports landmark curves, allowing points to slide along feature lines without explicit parameterization;
    \item The formulation is fully intrinsic, but preserves orientation;
    \item No combinatorial modifications to the input meshes are needed;
    \item The implementation relies solely on widely used operators from the tangent vector field processing toolbox. Leveraging the tensor product structure, we avoid constructing dense operators on the full product space.
\end{itemize}

\begin{figure}
    \includegraphics{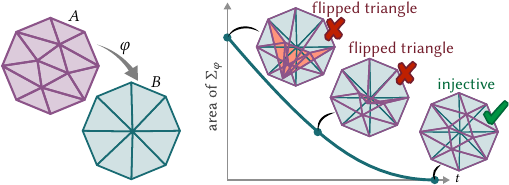}
    \caption{As in the case of curves, a map \(\varphi\) between surfaces \(A\) and \(B\) defines a surface \(\Sigma_\varphi\) in the product space whose area encodes the distortion of \(\varphi\). In particular, maps which flip triangles have larger surface area.
    \label{fig:minimal_surface_4d}}
\end{figure}

\subsection{Area Minimization and Metric Distortion}
\label{sec:AreaMinimizationAndMetricDistortion}

We now give a brief analysis of our distortion measure, the area of the ``graph'' \(\Sigma_\varphi\) associated to a map \(\varphi\) (\cref{eq:GraphDefinition}). In \cref{sec:RiemannianGeometryInTheProductOfSurfaces}, we express the area of \(\Sigma_\varphi\) using the singular values \(\sigma_1, \sigma_2\) of \(d\varphi\):
\begin{align}
    \textup{Area}(\Sigma_{\varphi}) = \int_{A} \sqrt{(1+\sigma_1^2)(1+\sigma_2^2)} \vol_{A}.
\end{align}
This area based distortion measure is bounded by the  Dirichlet energy\footnote{\(f^{\Area}(\sigma) \leq 1 + \tfrac12(\sigma_1^2 + \sigma_2^2)\) follows from the elementary inequality \(xy\leq \tfrac12(x^2+y^2)\)} and the area distortion\footnote{\(f^{\Area}(\sigma) \geq 1 + |\sigma_1\sigma_2|\) follows from the elementary inequality \((1+x^2)(1+y^2)\geq (1+|xy|)^2\)} of \(\varphi\), with equality (up to the addition of a constant) if and only if \(\varphi\) is conformal:
\begin{align}
    \int_{A}|\det d\varphi|~\vol_{A} \leq \Area(\Sigma_{\varphi}) - \Area(A) \leq \int_{A}\tfrac12|d\varphi|^2~\vol_{A}.
    \label{eq:DirichletEnergyBound}
\end{align}
\begin{definition}
    A bijective correspondence \(\varphi: A\to B\) is an \emph{area minimizing correspondence} if \(\Sigma_{\varphi}\) is a minimal surface.
\end{definition}
The bound in~\eqref{DirichletEnergyBound} also shows that maps with low total harmonic energy~\cite[Eq. 5]{Ezuz:2019:RHM} have small area, as \(\smash{\Sigma_{\varphi} = \Sigma_{\varphi^{-1}}}\), and thus a similar bound involving the area of \(B\) and the Dirichlet energy of \(\varphi^{-1}\) holds. Averaging the two inequalities shows that the area of \(\Sigma_{\varphi}\) is controlled by the sum of the Dirichlet energy of \(\varphi\) and its inverse.
Other distortion energies, that may also be related, are listed in~\cite[Table 1]{abulnaga2023symmetric} and \cite{poya2023geometric}. It would also be interesting to relate area minimizing correspondences to Gromov-Wasserstein distances~\cite{mandad2017variance}.

\subsection{Bijective Correspondences and Homology}
\label{sec:BijectiveCorrespondencesAndHomology}
We conclude with a brief comment on the constraints that we will later impose on the mapping surface \(\Sigma_\varphi\). Topological properties of \(\varphi\), like orientation preservation and bijectivity, are reflected in the topology of \(\Sigma_{\varphi}\). For instance, if we intersect \(\Sigma_\varphi\) with a ``vertical'' slice at a fixed \(a\in A\), or a ``horizontal'' slice at a fixed \(b\in B\),
\[
    \Sigma_{\varphi} \cap (\{a\}\times B) = \{(a,\varphi(a))\},\quad \Sigma_{\varphi}\cap(A\times\{b\}) = \{(\varphi^{-1}(b), b)\},
\]
we find a single intersection point. In particular, the projection maps \(\pi_{A}:\Sigma_{\varphi}\to A\) and \(\pi_{B}:\Sigma_{\varphi}\to B\) have degree \(\pm1\). The orientation preservation of \(\varphi\) is encoded in the sign of the degree, which in turn is encoded in the homology class of \(\Sigma_\varphi\). Thus we see that the homology class \([\Sigma_{\varphi}]\in H_2(A\times B)\) is not arbitrary---constant maps into \(B\), for example, induce graphs that are in the same homology class of \(A\) but do not satisfy the correct degree constraint for \(\pi_{B}\).

The homology class of surfaces \(\Sigma_\varphi\) arising from bijections \(\varphi\) is described explicitly in~\citet[Theorem 11.11, p.128]{Milnor:1974:CC} and is known as the \emph{diagonal homology class} \([\Delta]\) in \(A\times B\) since under the identification \(A\cong B\) induced by \(\varphi\), the graph \(\Sigma_{\varphi}\) is the diagonal in \(A\times A\). When \(A\) and \(B\) are simply-connected, the intersection of \(\Sigma_\varphi\) with horizontal and vertical slices completely determines its homology class. In \cref{sec:EncodingTheCorrectHomologyClass} we explain how we constrain the topology of \(\Sigma_\varphi\) to agree with \([\Delta]\) on slices.

\section{Related Work}
\label{sec:RelatedWork}

Computing mappings between surfaces is a long-standing challenge, and many algorithms have been proposed to address it.
Our method is distinguished from existing alternatives by a combination of properties: it minimizes isometric distortion while guaranteeing orientation preservation, though it does not strictly enforce bijectivity on discrete domains. Below we contrast it to several alternative families of approaches, but a comprehensive review lies beyond the scope of this paper---see the surveys by \citet{van2011survey,tam2012registration} and \citet{sahilliouglu2020recent} for broader context.
%These methods use diverse map representations, each yielding distinct properties such as bijectivity, continuity, or orientation preservation.

\subsection{Relaxed Map Representations}
% The quality of a map, including bijectivity, continuity, and orientation preservation, varies significantly with its representation.

\paragraph{Registration}
One approach to surface mapping is to deform one surface to match the other. While these methods simplify distortion analysis~\cite{beg2005computing,sharf2006snappaste,huang2008non,li2008global,tam2012registration,eisenberger2019divergence,eisenberger2020smooth}, they offer limited guarantees of bijectivity. Moreover, obtaining a close fit between the surfaces can involve expensive optimization, and aligning shape features often requires user intervention.

\citet{ezuz2019elastic,Ezuz:2019:RHM} minimizes the harmonic energy using vertex coordinates as variables. Unlike our method, their mapping remains undefined within triangles. However, like ours, it supports arbitrary initialization and can untangle correspondences.

\paragraph{Higher-Dimensional Relaxation}
Many works relax the bijectivity constraint by embedding the problem in higher-dimensional spaces, though this may introduce discontinuities, suboptimal distortion minimization, or orientation violations.
Kantorovich's optimal transport formulation replaces bijections with probability measures, sacrificing bijectivity and continuity. \citet{mandad2017variance} and \citet{Brifault:2025:OT} mitigate the lack of continuity by regularizing the transport plan.
Functional maps~\cite{ovsjanikov2012functional,ovsjanikov2016computing} represent mappings as linear operators acting on function spaces. However, converting these to continuous, bijective maps remains challenging~\cite{melzi2019zoomout,ren2021discrete}, and orientation preservation is non-trivial in the presence of intrinsic symmetries~\cite{ren2018continuous,ren2020maptree,donati2022complex}. Unlike exact bijective representations, functional maps enable partial mappings~\cite{litany2016puzzles}. \citet{solomon2012soft} frames correspondences as probabilistic measures, but it is subject to the same limitations.
\citet{vestner2017product} share our premise of computing maps as surfaces in a product space, noting that minimal surfaces reduce distortion---though without formal characterization for surfaces. Their method approximates $\Sigma_{\varphi}$ via kernel density estimation, reducing the problem to linear assignment. Consequently, their mapping is vertex-defined---with no guarantee of orientation preservation---and relies on search over the space of all possible vertex matchings.

\begin{figure}[!tpb]
    \includegraphics[width=\linewidth]{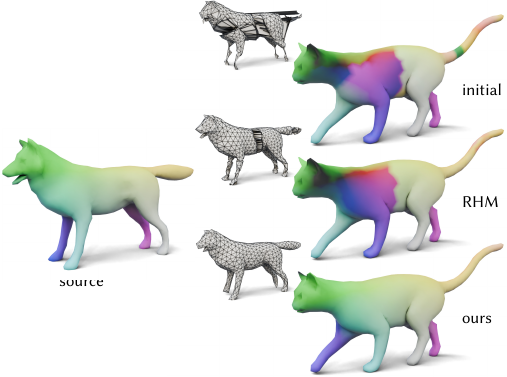}
    \caption{The maps computed by our algorithm are constrained to preserve normal orientation of the base surface. Intrinsic algorithms (such as functional maps) often fail to distinguish between symmetric parts. Here, we initialize with a map computed via \cite{ren2018continuous} that incorrectly swaps the left and right front legs. While RHM~\cite{Ezuz:2019:RHM}  improves smoothness, it fails to resolve the orientation reversal. In contrast, our approach successfully untangles the map, visualized by both color and geometry transfer.\label{fig:FmapsUntangling}}
\end{figure}

\subsection{Strictly Bijective Maps}

\paragraph{Mapping via Common Domains}
A provably bijective map can be obtained by composing two bijective maps from surfaces into a common reference domain. Methods leveraging Tutte embeddings ensure bijectivity by mapping into simple reference domains: the plane~\cite{kanai19973d,litke2005image,aigerman2014lifted,aigerman2015seamless}, periodic tilings~\cite{aigerman2015orbifold}, the sphere~\cite{aigerman2017spherical,Baden:2018:MR}, or hyperbolic space~\cite{tsui2013globally,aigerman2016hyperbolic,shi2016hyperbolic}. But distortion is typically controlled only for the surface-to-domain mapping, except for conformal maps~\cite{li2008globally,Baden:2018:MR} whose composition preserves conformality.
Other work has attempted to address these limitations by mapping to geometrically closer domains~\cite{schreiner2004inter}, blending multiple conformal maps~\cite{kim2011blended}, or computing mesh intersections in the common domain to optimize distortion~\cite{schmidt2019distortion,schmidt2020inter}. The latter achieves high accuracy but is computationally intensive and often converges to suboptimal solutions.
\citet{morreale2021neural} proposed using neural networks to encode surfaces, which alleviates some of the piecewise-linear limitation of existing mesh-based methods, but their method remains limited to pairs of surfaces which are both mapped into a common domain.

\paragraph{Remeshing}
A straightforward approach to encode correspondences is to approximate both surfaces with a mesh of shared connectivity. Distortion minimization then reduces to a remeshing problem, often solved via coarse-to-fine refinement~\cite{michikawa2001multiresolution,peng2016fast}. Early methods prioritized mesh approximation error~\cite{kraevoy2004cross,yang2020error}, while later works explicitly minimized distortion using volumetric mappings~\cite{yang2018volume}, intrinsic remeshing~\cite{takayama2022compatible}, or spherical parameterizations~\cite{schmidt2023surface}. These methods guarantee bijectivity and orientation preservation but require bijective initialization, precluding untangling.

\paragraph{Explicit Product Space} Another line of work enforces geometric consistency by explicitly searching for a submanifold in the discrete product space. \citet{windheuserGeometricallyConsistentElastic2011, windheuserLargeScaleIntegerLinear2011} initially formulated the search for an orientation-preserving diffeomorphism as an integer program on the product triangulation.
Recent work have attempted to improve the scalability of this approach with novel formulations and solvers~\cite{roetzerScalableCombinatorialSolver2022,amraniHighResolution3DShape2025,roetzerFastGloballyOptimal2025}.
However, these methods fundamentally need to define the connectivity of the product graph with massive constraint matrices, leading to prohibitive memory requirements. The resulting map is furthermore only combinatorial and doesn't define a continuous mapping within the faces.

In contrast, our method never explicitly meshes the product space, even though the implicit formulation encodes the geometric consistency constraints in the topology of the field's connection.

\begin{figure}[!tpb]
    \includegraphics{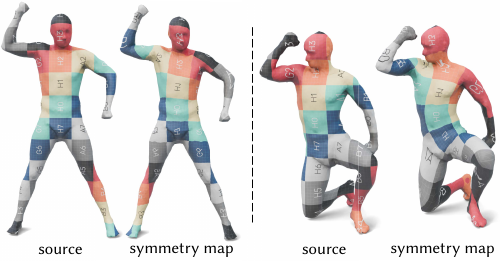}
    \caption{Our method preserves mapping orientation by encoding surface orientation in the connection form. This enables the computation of a symmetric map from an object to itself by simply inverting the normals of the target domain. \label{fig:SymmetryMap}}
\end{figure}

\subsection{Surface Reconstruction in Higher Dimensions}
Our algorithm builds a 2-dimensional surface embedded in a 4-dimensional space. \citet{kohlbrenner2023poisson} address the broader challenge of reconstructing a manifold from point samples in higher dimensions, but our setting is dramatically simplified because the surface of interest is defined directly as the zero set of an implicit function---not approximated from a sparse points set.

\subsection{Minimal Surfaces}
In the mathematical literature, the problem of computing minimal surfaces is often relaxed to computing minimal ``currents'' instead, which yields a convex relaxation of the problem in the space of generalized surfaces defined by \citet{Federer:1960:NIC}. For a computational introduction to the theory, see \citet{Wang:2021:CMS}. This approach has been used in computer graphics in applications ranging from neural surface modeling \cite{Palmer:2022:DC} to quad meshing \cite{Palmer:2024:LDF}. Unfortunately, when dealing with surfaces of codimension greater than 1, this relaxation is no longer tight. Indeed, even when considering two-dimensional surfaces bounded by a one-dimensional loop in \(\RR^4\), the relaxation to currents can already produce currents which have a lower area than any valid surface \cite{Young:1963:SEQ,DeLellis:2014:RMS}.

In the past twenty years, zero sets of critical points of the Ginzburg-Landau energy have been shown to be related to codimension-2 minimal surfaces~\cite{Jerrard:2002:JGL,Alberti:2005:VCF,DePhilippis:2022:NDM,Canevari:2023:YMH}, with very recent work showing that a generalization of the Ginzburg-Landau energy known as the \emph{self-dual Yang-Mills-Higgs functional}\footnote{it additionally includes a connection as a free variable of the optimization, penalizing the \(L^2\)-norm of its curvature in addition to the usual Ginzburg-Landau energy} provides a phase-field approximation of the codimension-2 area functional~\cite{Pigati:2021:MSA,Parise:2024:CSD}. These results are formulated in the language of complex line bundles, which provides the tools to generalize codimension-1 implicit surfaces represented as the level sets of real valued functions to codimension-2 implicit surfaces represented as the zero sets of complex sections with prescribed topology. This makes complex line bundles the perfect tool for our problem.

\section{Background}
\label{sec:Background}
\begin{figure}
\includegraphics{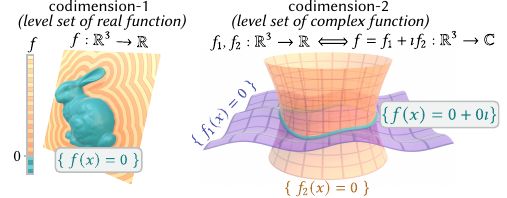}
\caption{\figloc{Left:} A codimension-1 object (like a 2D surface in 3D space) can be encoded as the zero level set of a real-valued function \(f\). \figloc{Right:} A codimension-2 object (like a 1D curve in 3D space, or a 2D surface in 4D space) can be encoded as the shared zero level set of a pair of real functions \(f_1, f_2\). Equivalently, a codimension-2 object is the zero level set of a single \emph{complex} function \(f(x) = f_1(x) + \imath f_2(x)\). This simple idea forms the basis of our implicit complex line bundle encoding (\cref{sec:ImplicitRepresentationViaComplexLineBundles}).}\label{fig:ImplicitRepresentations}
\end{figure}

% \begin{figure}
% \centering
% \includegraphics[width=.5\columnwidth]{notation.jpg}
% \caption{Notation}
% \label{fig:Notation}
% \end{figure}

In \cref{sec:ImplicitRepresentationViaComplexLineBundles,sec:ImplicitAreaMinimization} we explain more background about the underlying smooth mathematical theory, and we review the discretizations that we use in \cref{sec:CellComplexesAndProductMeshes,sec:DiscreteComplexLineBundles}. But the practically-minded reader can jump to \cref{sec:Algorithm} for a concrete description of our algorithm.

\subsection{Implicit Representation via Complex Line Bundles}
\label{sec:ImplicitRepresentationViaComplexLineBundles}
Our algorithm searches for a two-dimensional surface belonging to a specified homology class inside of a four-dimensional space (\cref{fig:ImplicitRepresentations}). Here we describe how these surfaces can be represented using mathematical objects known as \emph{complex line bundles}.

\paragraph{Warmup}
As a simpler example, consider a codimension-two subsets of a surface \(S\), \ie{} sets of oriented points. We could try to represent these points as the shared zero level set of a pair of functions \(f_1, f_2 : S \to \RR\), or equivalently as the zero level set of a single complex function \(f(x) = f_1(x) + \imath f_2(x)\). But not every collection of points can be represented as such a zero set. For instance, if \(f\) is a projection
\begin{wrapfigure}{r}{63pt}
    \includegraphics{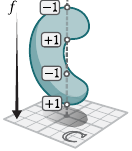}
\end{wrapfigure}
onto the \(xy\)-plane then its zero level set is the set of signed intersections between the surface and a ray shot up from the origin. As shown in the inset, such a point set will always have even size, and adding up the sign of each intersection point always yields a sum of 0. Indeed, the same is true for all complex functions, not just projections: for any smooth function \(f : S \to \CC\), the signs of the zeros sum to 0.

\begin{wrapfigure}{r}{45pt}
    \includegraphics{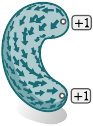}
\end{wrapfigure}
Alternatively, we could encode our point set as the zeros of a vector field. The Poincaré-Hopf theorem then guarantees that the indices of the zeros must sum to the Euler characteristic \(\chi(S)\), providing a different constraint on our zero set. So even though any vector field can \emph{locally} be represented as a complex function, the space of vector fields and the space of complex functions have different global structures---and different constraints on their zeros.

The space of vector fields and space of complex functions and are both examples of \emph{complex line bundles}. Formally, a complex line bundle on a manifold \(M\) is a space which associates a copy \(\CC_x\) of the complex plane to each point \(x\!\in\!M\), and locally looks like the product \(M \times \CC\). But its global topological structure may be different.

The analogue of a vector field on a general complex line bundle is a smooth mapping sending points \(x \in M\) to complex values \(z(x) \in \CC_x\). Such mappings are known as \emph{sections}. If \(M\) is a \(d\)-dimensional manifold, then the zero set of a section is an an oriented submanifold of dimension \(d-2\). Just as the zeros of a complex function or vector field sum to a fixed constant, the zeros of any section of a complex line bundle on \(M\) always lie in a particular \emph{homology class} of \(H_{d-2}(M;\ZZ)\), determined by the bundle's curvature.
%Conversely, for any class in \(H_{d-2}(M;\ZZ)\), one can construct a complex line bundle whose the zero sets lie in the given homology class.

\begin{wrapfigure}[6]{r}{70pt}
\vspace{-.5\baselineskip}
\includegraphics{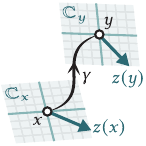}
\end{wrapfigure}
\paragraph{Connections and Curvature}
Sections can be studied using \emph{connections}, just as in the case of vector fields. A connection \(\nabla\) on a complex line bundle is a differential operator which provides a way of ``parallel transporting'' a value \(z(x) \in \CC_x\) along a path \(\gamma\) from \(x\) to \(y\) to obtain a value \(z(y) \in \CC_y\). However, parallel transport along different paths can result in different values in \(\CC_y\). Indeed, parallel transporting \(z(x)\) along a closed loop may produce a different value \(z\) upon returning to \(x\). The failure to close up is quantified by the curvature 2-form \(\Omega^\nabla\). See \cite[Appendix C]{Milnor:1974:CC} for more details.

\paragraph{Chern Classes}
The curvature 2-form is particularly important because it characterizes the topological class of the line bundle---and thus the homology class of our zero sets. The first Chern class of a complex line bundle is the cohomology class \({c_1 := \smash{[\tfrac{1}{2\pi} \Omega^\nabla]} \in H^2(M)}\). It encodes the homology class of our zero sets in the following sense: for any closed \(2\)-dimensional subset \(S \subseteq M\), the net number of signed intersections with the \((d\!-\!2)\)-dimensional zero set of a generic section is precisely \(\tfrac{1}{2\pi}\int_S \Omega^\nabla\). Remarkably, these values are always integers, and do not depend on the specific section or connection that we started from. The curvature of \emph{any} connection on the complex line bundle always lies in the same cohomology class. So by making an appropriate choice of complex line bundle structure we can control the topological class of the zero set of any section and can thus enforce the homology constraints of \cref{sec:BijectiveCorrespondencesAndHomology} by construction.

\paragraph{Connection Laplacians}
Just as the ordinary Laplacian measures the smoothness of scalar functions, connections can also be used to measure the smoothness of a section \(z\) on a complex line bundle. The connection Dirichlet energy
\begin{align}
    \DirichletEnergy(z) = \tfrac12\int_{M}|\nabla z|^2~\vol_{M}
\end{align}
then gives rise to the \emph{connection Laplacian} \(L^{\nabla}\) as the gradient of this Dirichlet energy.
When \(M = A\times B\) is the product of two surfaces \(A\) and \(B\), the Dirichlet energy can be expressed in a factorized form
\begin{align}
    \DirichletEnergy(z) = \!\int_{A}\left(\int_{B}|\nabla z|_{B}^2~\vol_{B}\right)\vol_{A} +\!\int_{B}\left(\int_{A}|\nabla z|_{A}^2~\vol_{A}\right)\vol_{B},
    \label{eq:DirichletEnergyDecomposition}
\end{align}
where \(|\nabla z|_{A}^2\) (\resp{}, \(|\nabla z|^2_{B}\)) denotes the norm of the covariant derivative of the section restricted to \(A\) (\resp{}, \(B\)).

\paragraph{Zero Set Distribution}
The placement of zeros of a smooth section in a complex line bundle is very closely related to the curvature of the connection. Experiments show that in low energy states of the connection Laplacian, zeros concentrate in regions of high curvature~\cite{knoppel2013globally,Weissmann:2014:SRS}. Recent theoretical work supports this intuition by showing that when you start with random sections whose zeros are distributed according to the curvature, the corresponding heat flow will connect these random zeros to the ground state zero set~\cite{Nicolaescu:2017:GBC}.
We use this intuition to design line bundles where curvature directs the zero set, producing high quality and controllable correspondences.

\subsection{Implicit Area Minimization}
\label{sec:ImplicitAreaMinimization}

The Ginzburg-Landau functional provides a way of minimizing the area of a surface encoded via a complex section, reducing the problem of bijective mapping to the computation of smooth sections on the product mesh. We briefly describe the smooth theory here, and later provide the discrete algorithm in \cref{sec:DiscreteGinzburgLandauMinimization}.

\paragraph{Allen-Cahn Energy}
Before introducing the full functional used in our method, we review the main ideas in the simpler codimension-1 case. Here a surface is represented as the zero set of a real-valued function \(u : M\to \RR\) and we consider the Allen-Cahn energy
\begin{equation}
    \mathcal{AC}_{\varepsilon}(u) := \int_{M}\tfrac{1}{2}|du|^2 + \frac{1}{4\varepsilon^2}(1 - u^2)^2\vol_{M}.
\end{equation}
The main contribution of the energy comes from the double well potential on the right which is minimized when \(u \equiv \pm 1\). The Dirichlet energy, on the other hand, prevents discontinuous jumps from a region where \(u\equiv +1\) to a region where \(u \equiv -1\). However, in the limit as \(\varepsilon \to 0\) the minimizers \(u_{\varepsilon}\) converge to function that has a jump across a minimal surface \cite{Modica:1977:USG}.

\paragraph{Ginzburg-Landau Energy}
The Ginzburg-Landau functional on a complex line bundle closely resembles the Allen-Cahn equation
\begin{align}
    \mathcal{GL}_{\varepsilon}(z) := \int_{M}\tfrac{1}{2}|\nabla z|^2 + \frac{1}{4\varepsilon^2}(1-|z|^2)^2\vol_{M},
    \label{eq:GinzburgLandau}
\end{align}
%
% \begin{wrapfigure}{r}{45pt}
%     \includegraphics{circular-well-potential.pdf}
% \end{wrapfigure}
\noindent{}The unit norm penalty is now minimized when \(z\in S^1\), and so we call it the \emph{circular well potential}. In the \(\varepsilon\to 0^+\) limit minimizers produce \(S^1\)-valued harmonic maps away from a codimension-2 set of zeros. Compared to the scalar valued case, the relationship between the Ginzburg-Landau functional and the area of its zero set is more subtle~\cite{Parise:2024:CSD}, but in certain situations the zeros of minimizers are known to form minimal surfaces~\cite{Lin:1999:CGL,Canevari:2023:YMH}. The intuition to keep in mind is (1) that \(\varepsilon\) controls the interface width between the zero set and the set where \(|z|\approx 1\) and (2) since the energy blows up as \(\varepsilon\to 0^+\) the blow up should happen on a set that is area minimizing.
% In~\cref{sec:Algorithm}, we describe our algorithm for computing a correspondence as the zero set of a low energy state of the Ginzburg-Landau functional.

\subsection{Cell Complexes and Product Meshes}
\label{sec:CellComplexesAndProductMeshes}
In our setting, the domain $M=A\times B$ is the product of two surfaces. This section introduces the construction of a discrete product space, and the next section presents the discretization of complex line bundles required to evaluate the discrete Ginzburg-Landau energy.

Mathematically, a \(d\)-dimensional mesh is an object called a \emph{cell complex}. Just as a polygon mesh is made up of zero-dimensional vertices, one-dimensional edges, and two-dimensional faces, a cell complex \(\smash{M = (M^0, M^1, \ldots, M^d)}\) is made up of \(k\)-dimensional \(k\)-cells \(\smash{\sigma^k_i \in M^k}\) for all \(k = 0, \ldots, d\). When the index \(k\) is unnecessary we following the usual convention for meshes and call the vertex set \(V_M := M^0\), the edge set \(E_M := M^1\), and the face set \(F_M := M^2\).

The boundary \(\partial \sigma^k_i\) of a \(k\)-cell \(\sigma^k_i\) is a collection of oriented \((k-1)\)-cells, which we write as a formal sum \(\smash{\partial \sigma^k_i = \sum_j \sigma^{k-1}_j}\). The boundary operator can be written as a matrix \(\partial_k \in \ZZ^{M^{k-1} \times M^k}\) where \(\smash{\left(\partial_k\right)_{j, i}}\) is \(1\) if \(\smash{\sigma^{k-1}_j}\) appears in \(\smash{\partial \sigma^k_i}\) with positive orientation, \(-1\) if it appears with negative orientation, and zero otherwise. The transpose of these boundary operators \(\smash{\mathsf{d}_{k-1} := \partial_{k}^\top}\) are the discrete exterior derivatives in discrete exterior calculus~\cite{Desbrun:2005:DEC}, and we denote the space discrete differential \(k\)-forms by \(\Omega^k(M;\mathbb{R})\).

\begin{figure}
    \includegraphics{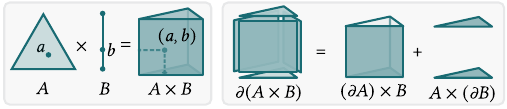}
    \caption{\figloc{Left}: Geometrically, the product space \(A \times B\) is the result of extruding \(A\) along \(B\). Algebraically, points of \(A \times B\) are pairs \((a, b)\) for \(a \in A\) and \(b \in B\). \figloc{Right}: The boundary operator obeys a product rule on product cells.}
    \label{fig:ProductSpace}
\end{figure}
\paragraph{Product Spaces}
Given two topological spaces \(A\) and \(B\), the product \(A \times B\) is the set of ordered pairs \((a, b)\) for \(a \in B\) and \(b \in B\) (\cref{fig:ProductSpace}). Similarly, if we have a pair of triangle meshes \(A = (V_A, E_A, F_A)\) and \(B = (V_B, E_B, F_B)\), their product space is a four-dimensional cell complex whose cells are products of cells from \(A\) and \(B\). Its vertex set is precisely the set \(V_{A \times B} = V_A \times V_B\) of pairs of vertices, while its edge set is \(E_{A \times B} = E_A \times V_B \cup V_A \times E_B\), its face set is \(F_{A \times B} = V_A \times F_B \cup E_A \times E_B \cup F_A \times V_B\). In general its set of \(k\)-cells is
\begin{equation}
    (A \times B)^k = \bigcup_{i+j=k} A^i \times B^j.
\end{equation}
The boundary operator on \(A \times B\) obeys a product rule (\cref{fig:ProductSpace})
\begin{equation}
    \partial \big(\sigma^k_i \times \sigma^l_j\big) = \big(\partial \sigma^k_i\big) \times \sigma^l_j + (-1)^k \sigma^k_i \times \big(\partial \sigma^l_j\big),
\end{equation}
which we can write in matrix notation as
\begin{equation}
    \partial^{A \times B}_{k} = \sum_{i+j=k}\partial^A_i \otimes \Id^B_j + (-1)^i\, \Id^A_i \otimes \partial^B_j.
\end{equation}

\subsection{Discrete Complex Line Bundles}
\label{sec:DiscreteComplexLineBundles}

\begin{wrapfigure}[4]{r}{63pt}
    \vspace{-1.5\baselineskip}
    \includegraphics{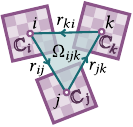}
\end{wrapfigure}
We use the discrete complex line bundles with connections defined by \citet{Knoppel:2016:CLB}. A discrete complex line bundle on a cell complex \(M = (V, E, F, \ldots)\) consists of:
\begin{enumerate}[leftmargin={7mm}]
\item a copy \(\CC_i\) of the complex plane for each vertex \(i \in V\),
\item a connection \(\connectionUnitary_\ij \in \CC\) of unit norm for each oriented edge \(\ij\in E\), satisfying \(\connectionUnitary_{\ji} = \connectionUnitary_{\ij}^{-1}\). These complex numbers act as discrete parallel transport maps \(\connectionUnitary_{\ij}:\CC_{i}\to\CC_{j}\),
\item a curvature \(\Omega_\ijk\!\in\!\RR\) for each face \(\ijk\in F\), that satisfies \(d\Omega = 0\) (trivially true on a surface), and is compatible with the connection: for each face \(\ijk\in F\), we have \(\connectionUnitary_{\ki}\connectionUnitary_{\jk}\connectionUnitary_{\ij} = \exp(\imath~\Omega_{\ijk})\).
\end{enumerate}
Just as in the smooth setting, the topological class of a discrete complex line bundle (and the implicit surfaces it can represent) are determined by the curvature 2-form \([\tfrac{1}{2\pi}\Omega] \in H^2(M)\). Since the parallel transport maps only determine \(\Omega\) modulo \(2\pi\), we have the freedom to modify the topological class of the bundle without changing the connection---in \cref{sec:ConstructingSurfaceConnections}, we use this ability to construct an appropriate bundle for our mapping problem.

Often, a discrete connection is represented a rotation angle \(\connectionAngle_{\ij}\) for each oriented edge, with the parallel transport maps \(\connectionUnitary_\ij = \smash{e^{\imath\connectionAngle_\ij}}\).

\paragraph{Sections and Zeros}
A discrete section \(z\) is an assignment of a complex number \(z_i \in \CC_i\) to each vertex \(i\). Given a section \(z\), the rotation of \(z\) along edge \(\ij\) is measured by the \emph{angular one-form}
\begin{equation}
    \fieldRotationAngle^z_{\ij} := \arg\!\left(\frac{z_j}{\connectionUnitary_{\ij} z_i}\right) \in [-\pi, \pi).
    \label{eq:AngularOneForm}
\end{equation}
If \(\smash{\fieldRotationAngle^z_{\ij} = -\pi}\), then \(z\) has a zero along edge \(\ij\). In the generic case where \(\fieldRotationAngle^z_{\ij} \in (-\pi, \pi)\), Knöppel and Pinkall define the index 2-form
\begin{equation}
    \textup{ind}^z := \tfrac{1}{2\pi}\big(d\fieldRotationAngle^z + \Omega\big), \label{eq:IndexForm}
\end{equation}
and prove that \(\smash{\textup{ind}^z_{\ijk}}\) is always an integer, which gives the sum of the indices of all zeros within face \(\ijk\). In contexts where the section \(z\) is clear, we drop the superscript and refer to \(\fieldRotationAngle_\ij\) and \(\textup{ind}_\ijk\).

Note that\([\textup{ind}^z] \equiv \smash{[\tfrac{1}{2\pi} \Omega]}\) as cohomology classes in \(H^2(M;\RR)\), as \(d\fieldRotationAngle^z\) is a closed discrete 2-form. Since \(\textup{ind}^z\) encodes the zeros of \(z\), we conclude that for any section \(z\) of the bundle, the zero set always lies in the same homology class---just as in the smooth setting.

\subsubsection{The Discrete Levi-Civita Connection}
\label{sec:DiscreteLeviCivitaConnection}
\begin{wrapfigure}{r}{58pt}
    \vspace{-.25\baselineskip}
    \includegraphics{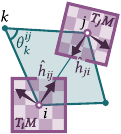}
\end{wrapfigure}
One important example of a discrete complex line bundle is the discrete Levi-Civita connection on the tangent bundle of a mesh \(M\), which~\citet{knoppel2013globally} construct by ``flattening'' a neighborhood of each vertex as follows: let \(\smash{\theta_i^{\jk}}\) be the corner angle of triangle \(\ijk\) at vertex \(i\), and let \(\smash{\Theta_i := \sum_{\ijk \succ i}\theta_i^{\jk}}\). Multiplying the angles around \(i\) by a factor of \(\tfrac{2\pi}{\Theta_i}\) yields scaled angles \(\smash{\tilde \theta_i^{\jk} := \tfrac{2\pi}{\Theta_i} \theta_i^{\jk}}\) that sum to \(2\pi\), allowing us to identify the tangent space of vertex \(i\) with the complex plane \(T_{i}M\cong \CC\). Then each halfedge \({\ij}\) leaving vertex \(i\) can be assigned a unit complex number \(\hat{h}_{\ij} \in T_{i}M\) giving its direction in this tangent space. The discrete Levi-Civita connection maps \(T_{i}M\) to \(T_{j}M\) while preserving the direction of edge \(\ij\), which can be written
\begin{equation}
\connectionUnitary^{LC}_{\ij} := -\hat{h}_{\ji}\hat{h}_{\ij}^{-1},
\label{eq:TangentBundleLeviCivitaRotations}
\end{equation}
after choosing any coordinates on \(T_{i}M\) and \(T_{j}M\).

To complete the complex line bundle, we assign curvature to each face \(\ijk\) based on its rescaled corner angles
\begin{equation}
\Omega^{LC}_{\ijk} = \tilde{\theta}_i^{\jk} + \tilde{\theta}_{j}^{\ki} + \tilde{\theta}_{k}^{\ij} - \pi.
\label{eq:TangentBundleLeviCivitaCurvatures}
\end{equation}
These curvatures are compatible with \(\connectionUnitary^{LC}\) and satisfy the Gauss-Bonnet theorem \(\sum \Omega^{LC}_{\ijk} = 2\pi \chi(M)\).

\subsubsection{Trivial Connections}
\label{sec:TrivialConnections}
Starting from any initial connection \(\connectionUnitary^0\) with curvature \(\Omega^0\), we can modify the curvature by solving a Poisson equation. If we are given another discrete 2-form \({\Omega}\in\Omega^2(M;\mathbb{R})\) on a simply connected triangle mesh \(M\) that has the same integral, \ie,
\begin{equation}
\sum_{\ijk\in F}\Omega^0_{\ijk} = \sum_{\ijk\in F}{\Omega}_{\ijk},
\end{equation}
we can find a discrete 1-form \(\alpha\in\Omega^1(M;\RR)\) that satisfies
\begin{equation}
\mathsf{d}_1\alpha = \Omega - \Omega^0.
\end{equation}
Using the Hodge decomposition, this amount to solving a Poisson equation on the dual mesh\footnote{letting \(\mathsf{L}^{-} = \mathsf{d}_1 *_1^{-1} \mathsf{d}_1^{\top}\) we have that \(\alpha = *_1^{-1}\mathsf{d}_1^\top\beta\) where \(\beta\) solves \(\mathsf{L}^{-}\beta = {\Omega} - \Omega^0\)}. The discrete connection
\begin{align}
    {\connectionUnitary}_{\ij} = e^{\imath\alpha_{\ij}}\connectionUnitary^0_{\ij}
    \label{eq:CurvatureTransformation}
\end{align}
is then compatible with the prescribed curvature 2-form \({\Omega}\). If we concentrate all of the prescribed curvature on a finite number of points this is the trivial connections algorithm of~\citet{Crane:2010:TCD}. We use this procedure in \cref{sec:ConstructingSurfaceConnections} to construct a connection with the correct curvature on each input surface, and in \secref{Initialization} to compute an initialization from a noisy correspondence.

\subsubsection{Finite Element Space} \label{sec:FiniteElementSpace}
To define continuous maps between triangle meshes, we need to be able to evaluate a section \(z\) not only at mesh vertices, but also inside triangles. To do so, we use the finite element space for complex line bundles on surfaces constructed by \citet{Knoppel:2016:CLB} and \citet{Liu:2016:Connection}. In our setting, the basis function for vertex \(i\) within triangle \(\ijk\) can be expressed as:
\begin{equation}
    \phi_i(p) = b_i(p) \exp\left[-\imath \left(\int_{p_i \rightarrow p} \rho_{\ijk} \right) \right],
    \label{eq:TriangleInterpolant}
\end{equation}
\begin{wrapfigure}{r}{53pt}
    \includegraphics{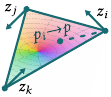}
\end{wrapfigure}
where \(b_i(p)\) is the barycentric coordinate of point \(p\) with respect to vertex \(i\), \(\rho_{\ijk}\) is the connection 1-form interpolated using Whitney interpolation~\cite{whitney1957geometric,Desbrun:2005:DEC}, and \(p_i \to p\) is the straight line connecting \(i\) to \(p\).

Using this finite element space on a triangle mesh \(S\), \citet[\S 6.1.1]{knoppel2013globally} give discretizations of the connection Laplacian \(\smash{L^\nabla_S \in \CC^{V \times V}}\) and the mass matrix \(\smash{M^\nabla_S \in \CC^{V \times V}}\), which depend on the parallel transport maps \(\connectionUnitary_\ij\) and face curvatures \(\Omega_\ijk\). We describe the full constructions in the supplementary material.

\begin{figure*}
\includegraphics[width=\linewidth]{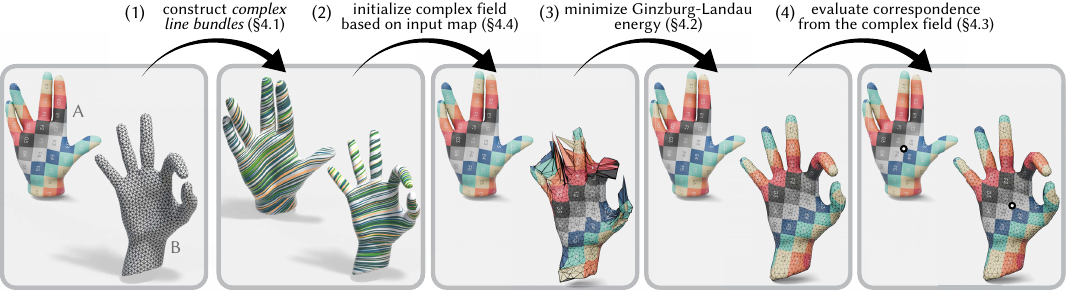}
\caption{Our algorithm proceeds in four steps. We begin with a pair of meshes \(A\) and \(B\) \figloc{(left)}. In step (1), we compute \emph{complex line bundle} structures on \(A\) and \(B\), which together define a complex line bundle structure on the product space \(A \times B\) (\cref{sec:ConstructingSurfaceConnections}). In step (2), we initialize our complex field based on a (potentially low quality) input map (\cref{sec:Initialization}). In step (3) we optimize the map by minimizing the discrete Ginzburg-Landau energy of our field (\cref{sec:DiscreteGinzburgLandauMinimization}). And finally in step (4) we evaluate the mapping at any desired points by identifying the zero set of \(z\) when interpolated via finite elements.}\label{fig:Pipeline}
\end{figure*}

\subsubsection{Product Space Finite Elements} \label{sec:ProductSpaceFiniteElements}
Recall that the complex field \(z: A \times B \to \mathbb{C}\) is defined at each vertex and stored as a complex matrix $Z$ of dimensions \(|V_A| \times |V_B|\). In order to interpolate these values over the 4-cells which make up the product space \(A \times B\), we define basis functions by taking tensor products of the basis functions \(\phi^A\) and \(\phi^B\) defined on surfaces \(A\) and \(B\) respectively.

Any 4-cell \(\smash{\sigma^4_{\ijk,abc} \in F_{A \times B}}\) is the product of triangles \(\ijk \in F_A\) and \(abc \in F_B\). The interpolated field at a point \((p, q) \in \sigma^4_{\ijk,abc}\) is:
\begin{equation}
    z(p, q) = \sum_{\substack{u \in \{i,j,k\} \\ v \in \{a,b,c\}}} Z_{u,v} \, \phi^A_u(p) \, \phi^B_v(q).
\end{equation}
Note that if we fix \(q\) and allow \(p\) to vary, we recover the standard surface interpolant on \(A\) up to a global rotation.

This interpolation scheme enables the definition of a connection Laplacian and the precise localization of zeros within each cell.

\subsubsection{Connection Laplacian}
\label{sec:ConnectionLaplacian}

We define a connection Laplacian on the product space by discretizing the Dirichlet energy (\cref{eq:DirichletEnergyDecomposition}). The discrete Dirichlet energy of a discrete section \(Z \in \mathbb{C}^{V_A \times V_B}\) on \(A\times B\) can be written in terms of left and right multiplication by the connection Laplacian and mass matrix defined on each manifold:
\begin{equation}
    \DirichletEnergy(Z) = \frac{1}{2} \langle Z, L_A^\nabla Z (M_B^\nabla)^\top \rangle_\CC + \frac{1}{2} \langle Z, M_A^\nabla Z (L_B^\nabla)^\top \rangle_\CC.
\end{equation}
Here we write \(\bullet^\top\) for the ordinary transpose, \(\bullet^\dagger\) for the conjugate transpose, and use the inner product \(\langle A, B\rangle_\CC = \operatorname{Re}\tr[A^\dagger B]\) for complex matrices. For a full derivation, see \cref{sec:TheProductSpaceDirichletEnergy}.

Similarly, the connection mass matrices \(M_A^\nabla \in \CC^{V_A \times V_A}, M_B^\nabla \in \CC^{V_B \times V_B}\) determine the \(L^2\) energy of product space sections via:
\begin{equation}
    \int_{A \times B} |z|^2 \vol_{A \times B} = \langle Z, M_A^\nabla Z (M_B^\nabla)^\top \rangle_\CC,
\end{equation}
and the scalar mass matrices \(M_A \in \RR^{V_A \times V_A}, M_B \in \RR^{V_B \times V_B}\) determine the \(L^2\) energy of real functions \(F \in \RR^{V_A \times V_B}\) on \(A \times B\):
\begin{equation}
    \int_{A \times B} f(a,b)^2 \vol_{A \times B} = \langle F, M_A F (M_B)^\top \rangle_\RR .
\end{equation}
% For simplicity, we take \(M_A\) and \(M_B\) to be the diagonal lumped mass matrices for triangle meshes \(A\) and \(B\).

\subsubsection{Encoding the Correct Homology Class}
\label{sec:EncodingTheCorrectHomologyClass}
In order to compute bijections between \(A\) and \(B\), we must prescribe the right curvature on the product space \(A \times B\).
Recall that the map \(\varphi: A \to B\) is defined implicitly: for any vertex \(i \in V_A\), the image \(\varphi(i)\) is the zero of the section \(z\) restricted to the slice \(\{i\} \times B\). This defines a section \(z^{(i)}\) on \(B\) where \(z^{(i)}(p) = z(i, p)\) for any point \(p \in B\).
Since vertex \(i\) should map to a single point on \(B\), this restricted section must contain exactly one zero.
Recalling that the total number of signed zeros on a the slice is \(\tfrac{1}{2\pi} \int \Omega\) (\cref{sec:ImplicitRepresentationViaComplexLineBundles}), this implies the curvature of the bundle restricted to any copy of \(B\) must integrate to \(2\pi\).
By symmetry, the curvature of the bundle restricted to any slice on \(A\) must also integrate to \(2\pi\).
To satisfy these constraints we start with two separate base connections on \(A\) and \(B\), with curvatures \(\Omega^A\) and \(\Omega^B\) each integrating to \(2\pi\), and we combine them to form a product connection on \(A \times B\).
The curvature associated to this resulting connection restricts to \(\Omega^A\) on any horizontal slice and to \(\Omega^B\) on any vertical slice, as desired.
Formally, when \(A\) and \(B\) are topological spheres, this product connection and curvature precisely encode the diagonal class \([\Delta]\), which contains all bijections (\cref{sec:BijectiveCorrespondencesAndHomology}).

\section{Algorithm}
\label{sec:Algorithm}
As illustrated in \cref{fig:Pipeline}, our algorithm for computing correspondences has four key steps:
\begin{enumerate}[leftmargin={5mm}]
    \item Design a discrete connection on \(A \times B\) such that each slice \(A\) and \(B\) of the product mesh contains one zero (\cref{sec:ConstructingSurfaceConnections}).
    \item Initialize the complex field \(Z\) from an input map (\cref{sec:Initialization}).
    \item Minimize the discrete Ginzburg–Landau energy of \(Z\) (\cref{sec:DiscreteGinzburgLandauMinimization}).
    \item Evaluate the correspondences by identifying the zero set of \(z\) after interpolation using a finite element basis (\cref{sec:EvaluatingTheCorrespondence}).
\end{enumerate}

To extend the algorithm's applicability, we introduce several straightforward modifications: support for surfaces with boundaries (\cref{sec:SurfacesWithBoundary}), support for point or curve landmarks (\cref{sec:Landmarks}), and a simple coarse-to-fine acceleration strategy (\cref{sec:MultiresolutionHierarchy}).

See the supplemental material for comprehensive pseudocode.

\paragraph{Notation}
Throughout, \(A = (V_A,E_A,F_A)\) and \(B=(V_B,E_B,F_B)\) are closed genus zero triangle meshes, with edge lengths \(\ell_{A}:E_{A}\to \mathbb{R}\) and \(\ell_{B}:E_{B}\to\mathbb{R}\). We write the discrete section on \(A \times B\) as a matrix \(Z \in \CC^{V_A \times V_B}\). We write \(\bullet^\top\) for the ordinary transpose, \(\bullet^\dagger\) for the conjugate transpose, and use the inner products \(\langle A, B\rangle_\CC = \operatorname{Re}\tr[A^\dagger B]\) for complex matrices, and \(\tr[A^\top B]\) for real matrices.
On each mesh \(S\) we let \(\smash{L^\nabla_S, M^\nabla_S \in \CC^{V_S \times V_S}}\) be the connection Laplacian and mass matrix for complex sections on \(S\) (\cref{sec:ConnectionLaplacian}), and let \(\smash{M_S \in \RR^{V_S \times V_S}}\) be the real diagonal mass matrix for scalar functions.

\subsection{Constructing Surface Connections}
\label{sec:ConstructingSurfaceConnections}
We begin by finding connections on \(A\) and \(B\) with total curvature \(2\pi\), which will allow us to find sections on \(A\) or \(B\) with a single zero---see \cref{sec:EncodingTheCorrectHomologyClass} for more discussion of the curvature constraint.

We can construct such a connection on each mesh by taking advantage of the \(2\pi\) ambiguity in the curvature of a discrete complex line bundle. We start with the trivial bundle \(\CC_i := \CC\) and trivial parallel transport maps \(\smash{\connectionUnitary^0_{\ij}:=1}\). We set the initial curvature \(\Omega^0\) to zero everywhere except on an arbitrarily chosen face \(f\in F\) where we set \(\smash{\Omega^0_{f}:=2\pi}\). This is a discretization of the skyscraper bundle~\cite[\S 7.2]{Knoppel:2020:RS}, which has the correct total curvature, but contains no geometric information. To obtain a more meaningful connection, we solve a linear system as in \cref{sec:TrivialConnections} to obtain a curvature of half the Gaussian curvature \(\smash{\Omega^{LC}_\ijk}\) on each face (\cref{eq:TangentBundleLeviCivitaCurvatures}), \ie{} we solve for a new connection \(\connectionUnitary_{\ij}\) with curvature \(\smash{\Omega_\ijk = \tfrac{1}{2} \Omega^{LC}_\ijk}\).

Together, these connections on \(A\) and \(B\) implicitly define a product connection on \(A \times B\). Fortunately, as we will see below, all of the relevant expressions factor into terms depending only on \(A\) or \(B\), so we do not have to assemble this full product-space connection.

\begin{remark}
    Any connection with the same curvature can in principle be used instead. In~\cref{app:DiscreteSurfaceConnections} we give two alternative constructions more directly related to tangent vectors and the Levi-Civita connection.
\end{remark}

\subsection{Discrete Ginzburg-Landau Minimization}
\label{sec:DiscreteGinzburgLandauMinimization}
We discretize the Ginzburg-Landau energy, \cref{eq:GinzburgLandau}, via a finite-element Dirichlet energy and a lumped discretization of the circular well penalty term. This energy can be minimized with L-BFGS using the expressions for the energy and gradient provided below.

\paragraph{Dirichlet Energy}
As explained in \cref{sec:ConnectionLaplacian}, the Dirichlet energy on \(A \times B\) can be assembled from the connection Laplacians and mass matrices of \(A\) and \(B\) as follows
\begin{equation}
\DirichletEnergy(Z) = \tfrac{1}{2}\left\langle Z \overline{M}^\nabla_B, L^\nabla_A Z \right\rangle_\CC + \tfrac{1}{2}\left\langle \overline{M}^\nabla_A Z,  Z L^\nabla_B\right\rangle_\CC.
\label{eq:ProductSpaceDirichletEnergy}
\end{equation}

\paragraph{Circular Well Potential}
We discretize the penalty term at the vertices, weighted by their barycentric dual volumes in \(A\times B\):
\begin{equation}
    \CircularWellPotential(Z) = \sum_{(i,j)\in V_A\times V_B} (1-|Z_{i,j}|^2)^2 M_{i,i}M_{j,j}.
\end{equation}
To write \(\CircularWellPotential\) in matrix notation, we define the norm deviation matrix \(U \in \RR^{V_A \times V_B}\) by \({U_{i,j} = \left|Z_{i,j}\right|^2 - 1}\), so that \(\CircularWellPotential(Z) = \left\langle U, M_A U M_B^\top\right\rangle_\RR\). This vertex-based discretization minimizes the stencil size---and computational cost---of function and gradient evaluations.

\paragraph{Ginzburg-Landau Energy}

The discrete energy of a section \(Z\) is
\begin{equation}
\begin{aligned}
    \DiscreteGinzburgLandau_\lambda(Z) &:= \DirichletEnergy(Z) + \tfrac{\lambda}{4}\CircularWellPotential(Z) \\
    & = \tfrac{1}{2}\left\langle Z\overline{M}^\nabla_B, L^\nabla_A Z \right\rangle_\CC + \tfrac{1}{2}\left\langle \overline{M}^\nabla_A Z, ZL^\nabla_B \right\rangle_\CC \\
    &\quad+ \tfrac{\lambda}{4} \left\langle U M_B, M_A U\right\rangle_\RR,
\end{aligned}
\label{eq:DiscreteGL}
\end{equation}
where the variable \(\lambda\) is called the \emph{Ginzburg-Landau parameter}.
Its gradient can also be written as the following matrix in \(\CC^{V_A \times V_B}\):
\begin{equation}
\nabla_Z \DiscreteGinzburgLandau_\lambda = L^\nabla_A Z \left(M^\nabla_B\right)^\top + M^\nabla_A Z \left(L^\nabla_B\right)^\top + \lambda \left(M_A U M_B^\top\right) \odot Z,
\label{eq:DiscreteGLGrad}
\end{equation}
where \(\odot\) denotes the element-wise product.

The choice of Ginzburg-Landau parameter significantly influences the quality of the correspondence. If \(\lambda\) is below the minimal eigenvalue of $L^{\nabla}_{A\times B}$ then the only critical point of the energy is $z\equiv 0$ (\cref{lem:global_mini_GL}). But if $\lambda$ is too large, the circular well potential repels zeros away from the vertices. And since the Dirichlet term carries a comparatively low weight, the field $z$ tends to remain close to its initial value during optimization.

% One simple strategy is to set \(\lambda\) based on the smallest eigenvalues \(\lambda_A, \lambda_B\) of \(L^\nabla_A\) and \(L^\nabla_B\), as \(\lambda_A + \lambda_B\) provides a tight approximation on the smallest eigenvalue of \(L^\nabla_{A \times B}\). Experimentally, we find $\lambda\approx 100 (\lambda_A + \lambda_B)$ provides near-optimal results across a variety of meshes (\cref{fig:OptimalLambda}).

One simple strategy is to set \(\lambda\) relative to the smallest eigenvalue $\lambda_{A \times B}$ of \(L^\nabla_{A \times B}\). In the smooth setting, this eigenvalue decomposes as \(\lambda_{A \times B} = \lambda_A + \lambda_B\), where \(\lambda_A\) and \(\lambda_B\) are the smallest eigenvalues of the connection Laplacians on \(A\) and \(B\).
While this equality only holds approximately in the discrete setting due to differing mass matrix discretizations (see \cref{sec:TheGinzburgLandauParameter} for details), it remains a good  estimate. Experimentally, we find $\lambda\approx 100 \lambda_{A\times B}$ provides near-optimal results across a variety of meshes (\cref{fig:OptimalLambda}).

For challenging initializations, starting from a smaller value of \(\lambda\) is often necessary to smooth the initial correspondences and match surface features correctly. For example, in \cref{fig:ComparisonInitialization,fig:FmapsUntangling}, the optimization is performed in two stages: first, with \(\lambda = 10 \lambda_{A\times B}\) to align large-scale features, and second, with \(\lambda = 100 \lambda_{A\times B}\) to refine the correspondences and address high-frequency matching.

\begin{figure}[!htpb]
    \includegraphics{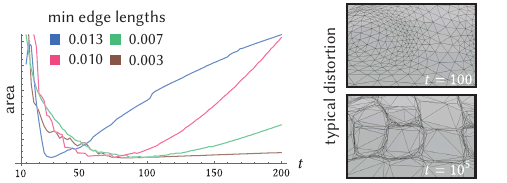}
    \caption{An experimental search for the optimal choice of Ginzburg-Landau parameter \(\lambda = t\,\lambda_0\) where \(\lambda_0 = \lambda_A + \lambda_B\).
    Since the area depends on the distortion between the surfaces, to consistently analyze the optimal parameter choice we normalize the area by the maximal area measured for each example separately.
    Large values of \(\lambda\) cannot produce smooth correspondences unless the mesh has a sufficiently high resolution, as indicated by having a small minimum edge length. \label{fig:OptimalLambda}}
\end{figure}

\subsection{Evaluating the Correspondence}
\label{sec:EvaluatingTheCorrespondence}

\begin{figure}
\includegraphics{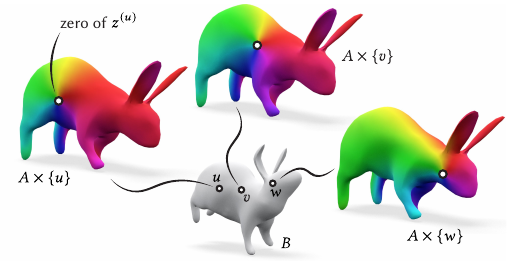}
\caption{To evaluate the map at a vertex \(u \in V_B\), we take the \(u\)'th column of \(Z\) as a section \(z^{(u)}\) on \(A\). We can visualize this section as a complex function in a chart that covers all but a single face of \(A\). Now the image \(\varphi(u)\) of vertex \(u\) is given by the location of the zero on \(A\).}
\label{fig:LocatingZeros}
\end{figure}

By construction (\cref{sec:ConstructingSurfaceConnections}), for any vertex \(v \in V_B\), the \(v\)-th column of \(Z\) is a section \(z^{(v)}\) on \(A\)--- simply denoted $z$ below---with a single zero (see \cref{fig:LocatingZeros}).
To locate this zero, we first compute the integer-valued index 2-form \(\ind^{z}\) using \cref{eq:IndexForm}. This form takes a non-zero value on exactly one face \(\ijk \in F_A\), where the zero lies.

We then compute the barycentric coordinates \((b_j, b_k)\) of the zero within \(\ijk\) (\cref{sec:FiniteElementSpace}) by solving the following system:
\begin{equation}
    (1-b_j-b_k) |\,z_i| + b_j |\,z_j| e^{\imath\left(b_k \Omega + \fieldRotationAngle_{\ij}\right)} + b_k |\,z_k| e^{-\imath\left(b_j \Omega + \fieldRotationAngle_{\ki}\right)} = 0.
    \label{eq:zero_system}
\end{equation}

\begin{figure}
\includegraphics{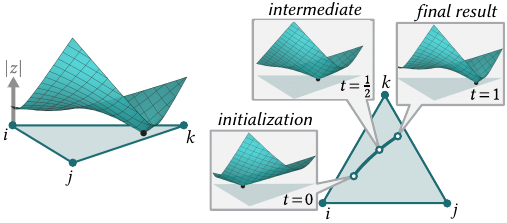}
\caption{\figloc{Left:} In order to evaluate the correspondence, we have to find a root of the interpolated section \(z\) within some triangle \(\ijk\). But the interpolant is nonlinear and non-convex, so solving directly with Newton's method might not find the desired root. \figloc{Right:} Instead, we interpolate the face curvature from flat at \(t=0\)---in which case the interpolant becomes linear---to the full curvature \(\Omega_\ijk\) at \(t=1\), solving with Newton's method at each time step.}\label{fig:ZeroInterpolation}
% if you want to edit this figure, the Mathematica source code is in /Media/TriangleZero
\end{figure}
Directly solving \cref{eq:zero_system} with Newton's method is unlikely to yield the solution within triangle $\ijk$. Instead, we adopt a homotopy continuation approach, and deform the geometry from a flat triangle (\(t\!=\!0\)) to the target curvature \(\Omega_{\ijk}\) (\(t\!=\!1\)), tracking the zero as it moves.
We use Newton's method to find the zero at each time step (\cref{fig:ZeroInterpolation}). Pseudocode is provided in the supplementary material.

We stop the zero from ``escaping'' by ensuring that the zero never hits an edge while interpolating, \ie{} $\fieldRotationAngle_{\ij}^z(t), \fieldRotationAngle_{\jk}^z(t), \fieldRotationAngle_{\ki}^z(t)\!\in\!(-\pi, \pi)$ for $t \in [0, 1)$. This is achieved via the following linear interpolation:
\begin{align}
    & \fieldRotationAngle_{\ij}^z(t) = \fieldRotationAngle_{\ij}^z + (1-t) \tfrac{1}{3} \left( \Omega_{\ijk} - 2\fieldRotationAngle_{\ij}^z + \fieldRotationAngle_{\jk}^z + \fieldRotationAngle_{\ki} \right) \\
    & \fieldRotationAngle_{\jk}^z(t) = \fieldRotationAngle_{\jk}^z + (1-t) \tfrac{1}{3} \left( \Omega_{\ijk} + \fieldRotationAngle_{\ij}^z - 2\fieldRotationAngle_{\jk}^z + \fieldRotationAngle_{\ki} \right) \\
    & \fieldRotationAngle_{\ki}^z(t) = \fieldRotationAngle_{\ki} ^z+ (1-t) \tfrac{1}{3} \left( \Omega_{\ijk} + \fieldRotationAngle_{\ij}^z + \fieldRotationAngle_{\jk}^z - 2\fieldRotationAngle_{\ki}^z \right) \\
    & \Omega_{\ijk}(t) = t \Omega_{\ijk}
\end{align}

One can check that $\fieldRotationAngle_{\ij}^z(t) \in (-\pi, \pi)$ for $0 \leq t < 1$ and that the index stays fixed, so the zero must remain inside triangle $\ijk$:
\begin{align}
    \forall t, \quad \fieldRotationAngle_{\ij}^z(t) + \fieldRotationAngle_{\jk}^z(t) + \fieldRotationAngle_{\ki}^z(t) + \Omega_{\ijk}(t) = 2\pi \ind^z_{\ijk}.
\end{align}

\paragraph{Edge-Edge Intersections}
The section encodes not only vertex locations, but also the entire embedding of the edges (\cref{fig:Teaser,fig:ZeroPlacement}). The intersection of an edge of \(A\) and an edge of \(B\) may be calculated by finding a zero inside an ``edge-edge'' face of the product space. Since these faces have zero curvature, the zero can be found by solving a single quadratic equation, which is presented in \cref{sec:EdgeEdgeIntersections}.

\paragraph{Evaluation Inside a Face} We can also evaluate the correspondence for any point \(p\) within triangle \(\ijk \in F_B\) using the basis functions in \cref{sec:FiniteElementSpace}.
Using the complex weights \(\phi_\bullet(p)\) for each vertex \(i, j, k\) of \(B\) given by \cref{eq:TriangleInterpolant}, we define a section \(z^{(p)}\) on \(A\) as the linear combination of the columns of \(Z\) corresponding to vertices \(i, j, k\). Solving for its zero as above yields the image of \(p\) on \(A\).

\begin{figure}[!htpb]
    \includegraphics{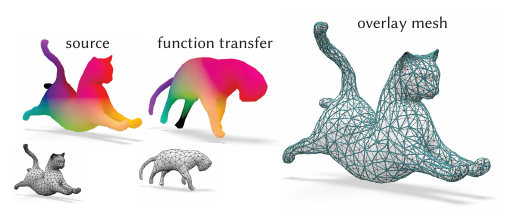}
    \caption{Our algorithm produces bijective and continuous mappings even when the meshes exhibit substantially different connectivity and sampling. The zero set of the field may place multiple vertices of \(A\) within the same face of mesh \(B\), while preserving low distortion and bijectivity.\label{fig:ZeroPlacement}}
\end{figure}

\subsection{Initialization}
\label{sec:Initialization}
Our method can start from a (potentially noisy) input function given as vertex-face maps \(\varphi : V_A \to F_B\) and \(\psi : V_B \to F_A\). To use these in our algorithm, we must convert them into an initial complex section \(z : A \times B \to \CC\) whose zero set approximates the graphs of \(\varphi\) and \(\psi\).

\paragraph{Strategy}
We use the principle that low-energy states of a connection Laplacian concentrate their zeros in regions of high curvature (\cref{sec:ImplicitRepresentationViaComplexLineBundles}). In particular, we construct a connection \(\connectionUnitary^{\varphi, \psi}\) on \(A \times B\) that concentrates all its curvature exactly along the graphs of \(\varphi\) and \(\psi\). We then take our initial section \(z\) to be the smallest eigenvector of the associated Laplacian. As this bundle lies in the same topological class as the bundle from \cref{sec:ConstructingSurfaceConnections}, this section serves as a valid initialization for the following Ginzburg-Landau minimization.

\paragraph{Construction}
We build \(\connectionUnitary^{\varphi, \psi}\) slice-by-slice using the trivial connections algorithm from \cref{sec:TrivialConnections}. For each vertex \(i_{A}\!\in\!V_A\), we compute a connection \(\connectionUnitary^{i_A}\) on \(B\) that concentrates \(2\pi\) curvature into the single face \(\varphi(i_A)\). Then, for every edge of the form \(i_A \times e_B\) in the product space, we assign \(\smash{\connectionUnitary^{\varphi, \psi}_{i_A \times e_B}} = \connectionUnitary^{i_A}_{e_B}\) Symmetrically, for each \(i_{B}\!\in\!V_B\), we run trivial connections on \(A\), concentrating curvature in face \(\psi(i_B)\) and assign this connection to product space edges of the form \(e_A \times i_B\). This fully defines a connection on the product space whose curvature is concentrated to the graph of the input maps. The smallest eigenvector of the associated Laplacian acts as a good representation of the input correspondences (see \cref{fig:NnInit}).

\paragraph{Laplacian}
The above connection does not factor as the product of independent connections on \(A\) and \(B\), so it would normally require assembling the full four-dimensional finite element Laplacian. Instead, we consider a simpler discretization of the Dirichlet energy using a lumped discretization for the outer integrals of \cref{eq:DirichletEnergyDecomposition} while retaining a finite-element discretization for the inner integrals over the 2D slices. Explicitly, for every vertex \(i_{A}\!\in\!V_A\),  let \(\smash{L^{\nabla,i_{A}}_{B}}\) be the finite element Laplacian  of \(\connectionUnitary^{i_{A}}\) on \(B\), and let \(z^{(i_A)}\) be the restriction of section \(Z\) to the slice. Using similar notation for the slices fixed at \(i_B\), we write the Dirichlet energy of the connection \(\connectionUnitary^{\varphi,\psi}\) as:
\begin{equation}
\begin{aligned}
    \DirichletEnergy_{\varphi,\psi}(Z) & = \sum_{v \in V_A}(M_A)_{v,v} \langle L^{\nabla,v}_{B} z^{(v)}, z^{(v)}\rangle_{\CC} \\
    & \qquad + \sum_{v \in V_B}(M_B)_{v,v} \langle L^{\nabla,v}_{A} z^{(v)}, z^{(v)}\rangle_{\CC},
\end{aligned}
\label{eq:SlicewiseLaplacian}
\end{equation}
and so the blocks in the product-space connection Laplacian are given by the surface connection Laplacians on the individual slices weighted by the dual area of the vertex on the complementary mesh.

The smallest eigenvalue can then be computed via an iterative matrix-free solver such as LOBPCG~\cite{Knyazev:2007:BLO}.

\paragraph{Initialization from Distributions}
Rather than starting from a map \(\varphi : V_A \to F_B\), one can start from \emph{distributions} on \(B\) associated to the vertices of \(A\), provided \eg{} by a functional map or optimal transport plan. The procedure is almost exactly the same: the only change is that for each vertex of \(A\) we spread curvature over \(B\) proportional to its density instead of concentrating all curvature in one face \(\varphi(i_A)\).

\begin{figure}[!htpb]
    \includegraphics{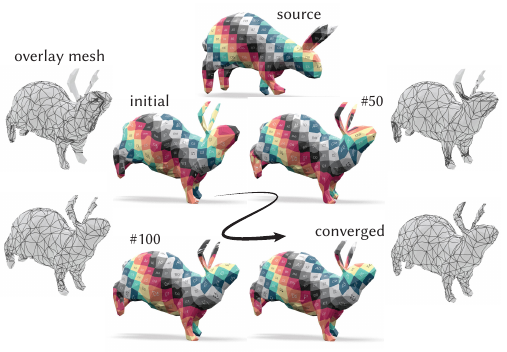}
    \caption{High quality correspondences can be computed even on very coarse tessellations. On this example, with \(|V_A|=474\) and \(|V_B|=502\), the solution is computed in 21 seconds. The robustness to tessellation density is a consequence of the finite-element representation of the implicit surface: the complex section \(z\) smoothly extends to the interior of the faces and does not require fine sampling to represent a continuous bijection. \label{fig:LowResolution}}
\end{figure}

\subsection{Multiresolution Hierarchy}
\label{sec:MultiresolutionHierarchy}
Our method can be trivially accelerated using a multiresolution hierarchy of meshes. We first compute an optimal complex section \(z^{(0)}\) on coarsened meshes \(A^{(0)}\) and \(B^{(0)}\) by minimizing the Ginzburg-Landau energy.
To transfer this solution to fine meshes \(A\) and \(B\), we distinguish between two cases depending on the bundle structure.

\paragraph{Direct Upsampling}
For the surface connections constructed in \cref{sec:ConstructingSurfaceConnections}, the fibers \(\CC_i\) are globally identified with the complex plane. This allows us to transfer the complex values of the section directly without explicit change of basis (see \cref{fig:MultiresolutionSection}). We upsample \(z^{(0)}\) by closest point projection: for every vertex pair \((i_A, i_B)\) in the fine product mesh \(A \times B\), we find the closest geometric points \(p_A \in A^{(0)}\) and \(p_B \in B^{(0)}\) and evaluate \(z^{(0)}\) at \((p_A, p_B)\), using interpolation defined in \cref{sec:ProductSpaceFiniteElements}. This upsampled section is then used as initialization for the fine-scale Ginzburg-Landau minimization, or for the iterative eigenvector solver. In practice, this can provide a $12\times$ speedup in computation time (\cref{sec:LimitationsAndFutureWork}). Note that for this transfer to be as accurate as possible, the face selected in \cref{sec:ConstructingSurfaceConnections}, which concentrates the curvature in the coarse mesh, should contain the corresponding face in the fine mesh.
\begin{figure}[!tpb]
    \includegraphics{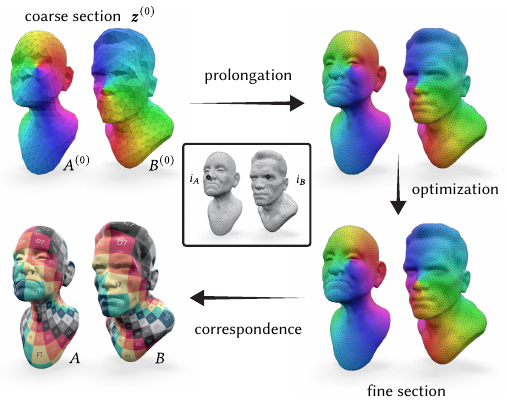}
    \caption{Our implicit representation of the map can be prolongated to a finer mesh by closest point interpolation of the complex section. Thus we can initialize our optimization from a coarse solution, which already provides a very good approximation of the desired map. To visualize the sections, we plot the complex values on the slices \(\{i_A\}\times B\) and \(A\times \{i_B\}\), and to visualize the correspondence we transfer a conformal parameterization from \(B\). \label{fig:MultiresolutionSection}}
\end{figure}

\paragraph{Geometric Initialization}
When using a more general connection, \eg{} based on tangent vector fields (\cref{app:DiscreteSurfaceConnections}), direct transfer of complex values might not be possible. In these cases, we instead transfer the map directly: we first extract the coarse correspondences \(\varphi^{(0)}\) and \(\psi^{(0)}\) from the section \(Z^{(0)}\), as described in \cref{sec:EvaluatingTheCorrespondence}. We then construct vertex-to-face maps between the fine meshes \(A\) and \(B\) by composing these coarse correspondences with closest point projections, and initialize the fine level as in \cref{sec:Initialization}.

\subsection{Landmarks}
\label{sec:Landmarks}
The Ginzburg-Landau energy can be modified to accommodate a user-specified set of landmarks \(\Landmarks = \{(l_i^A, l_i^B)\}_{i=1}^{p}\) (\cref{fig:Landmarks}). Since the map sends a point \(a \in A\) to \(b \in B\) if and only if \(z(a,b) = 0\), we can encourage this behavior by penalizing non-zero values at the landmarks. Hence, we consider the specifications as soft-constraints.

We replace the standard circular well potential with a spatially varying version, defined by a non-negative potential function \(V : A\times B\to \RR\) having isolated zeros at \(\Landmarks\):
\begin{equation}
    \CircularWellPotential(Z) := \tfrac14\int_{A\times B}(V(p) - |z(p)|^2)^2~\vol_{A\times B},
\end{equation}
We construct \(V\) using Gaussians centered at the landmarks:
\begin{equation}
    V(p) :=\hspace{-2.5mm}\min_{(p_A,p_B)\in\Landmarks}\Big(1 - \exp\big(-\!\tfrac{1}{2\sigma_A^2}d_A(p,p_A)^2 - \tfrac{1}{2\sigma_{B}^2}d_{B}(p,p_B)^2 \big) \Big),
\end{equation}
where \(\sigma_A,\sigma_B > 0\) determine the kernel width, and \(d_A(\cdot,\cdot)\) and \(d_B(\cdot,\cdot)\) denote the geodesic distance on \(A\) and \(B\), respectively. Finer granularity over the pinning potential can be achieved by specifying a landmark dependence choice of $\sigma_A$ and $\sigma_{B}$.

\begin{figure}[!htpb]
    \includegraphics{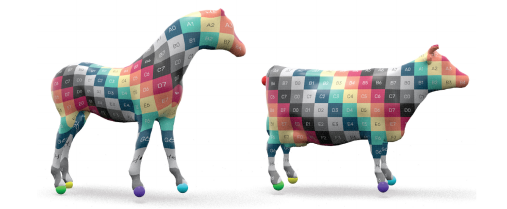}
    \caption{Landmarks are specified using a pinning potential that has local minima at the specified points in the product mesh.\label{fig:Landmarks}}
\end{figure}

This modification of the Ginzburg-Landau functional is known as \emph{singularity pinning} in the physics literature\footnote{it has been observed experimentally that the relevant singularities configurations are related to material defects that can be encoded in \(V\)}, and \citet{Aftalion:2001:PPG} proved that in the small \(\varepsilon\) limit that zeros of the Ginzburg-Landau minimizers (in a two-dimensional background) are ``pinned'' at the local minimizers of \(V\).
For this attraction to be effective in practice, the width of the Gaussian kernel must be sufficiently large. Since \(V\) modifies the energy gradient only locally, a zero far outside the well, where \(V\!\approx\! 1\), will not be affected by the potential.

The pinning potential can also be modified to support unparameterized curve to curve correspondences (\cref{fig:CurvePinning}). For a pair of curves $\gamma_{A}\subset A$ and $\gamma_{B}\subset B$ we modify $V$, and replace the point-to-point distances with point-to-curve distances $d_A(p,\gamma_{A})^2$ and $d_B(p,\gamma_{B})^2$. The resulting contribution encourages the section to vanish along the two-dimensional surface $\gamma_{A}\times \gamma_{B}\subset A\times B$.

\begin{figure}[!htpb]
    \includegraphics{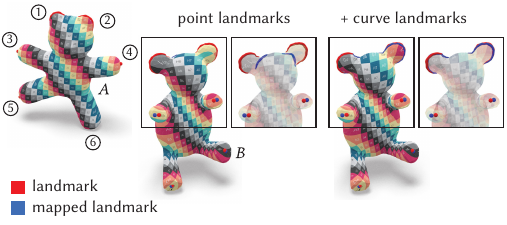}
    \caption{Enforcing curve to curve correspondences where points are allowed to slide along the specified curves is notoriously challenging. Our implicit representation can elegantly handle these as soft-constraints by modifying the circular well potential. The red points and curves are the specified landmarks, and the blue points and curves visualize their image in $B$ under the computed correspondence. \label{fig:CurvePinning}}
\end{figure}

\subsection{Surfaces with Boundary}
\label{sec:SurfacesWithBoundary}
In~\cref{fig:Boundary} we compute a map between two topological disks by filling in each boundary component to obtain topological spheres. To ensure that the boundaries are mapped correctly we also include a curve-to-curve singularity pinning potential for them. For simplicity we fill the boundaries with triangle fans, although a more isotropic mesh of the boundary disk would provide a better finite element space, and may be necessary for finely tessellated boundaries.

\begin{figure}[!htpb]
    \includegraphics{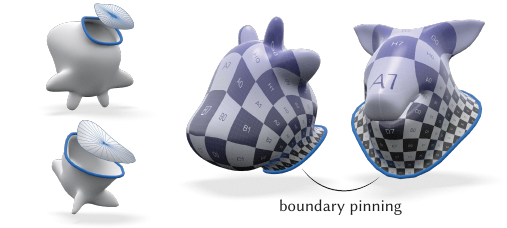}
    \caption{Our method handles surfaces with boundary by filling in a disk and pinning the boundaries together.}
    \label{fig:Boundary}
\end{figure}

\begin{figure}[!b]
    \includegraphics{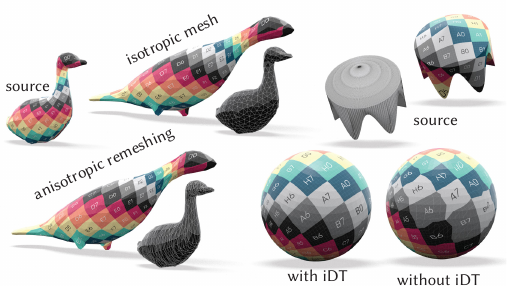}
    \caption{Our method is robust to triangulation quality, finding the same correspondence when the source is triangulated with an isotropic or a highly anisotropic mesh \figloc{(left)}. Using the intrinsic Delaunay triangulation (iDT) yields smooth correspondences even in the most extreme situations \figloc{(right)}. \label{fig:MeshIndependence}}
\end{figure}

\subsection{Intrinsic Triangulations}
\label{sec:IntrinsicTriangulations}
Since the Ginzburg-Landau functional is defined solely in terms of intrinsic quantities, we can improve accuracy by using the intrinsic Delaunay triangulation~\cite{Bobenko:2007:DLB}. While our discretization is already robust to triangulation quality and sampling density (\cref{fig:LowResolution}), using the intrinsic Delaunay triangulation produces higher quality maps even on severe examples (\cref{fig:MeshIndependence}). And intrinsic Delaunay triangulations improve the PDE-based geodesic distance approximation used to constrain landmarks. We represent intrinsic triangulations with the integer coordinates of~\citet{Gillespie:2021:ICI}, and applied this intrinsic preprocessing throughout. Intrinsic Delaunay refinement~\cite{Sharp:2019:NIT,Gillespie:2021:ICI} may offer additional improvements, especially on overly coarse inputs.

\section{Results and Evaluation}
\label{sec:Results}
In this section, we demonstrate the fundamental features and applications of our approach and compare with state-of-the-art techniques for computing bijective correspondences. Note that surface correspondences are general purpose tools, so our method has further applications beyond the basic operations we consider here (mesh transfer, surface interpolation, \etc).

\begin{figure}[!htpb]
    \includegraphics{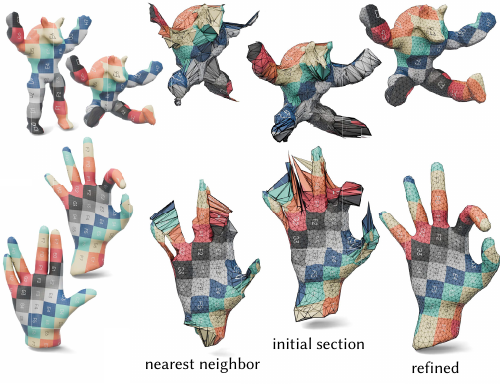}
    \caption{Correspondences obtained from closest-point initialization, visualized with texture transfer \figloc{(left)}. Our initialization algorithm finds an initial complex section whose zero set approximates the graph of the nearest neighbor correspondences. Minimizing the Ginzburg-Landau energy refines this section to produce a high-quality correspondence \figloc{(right)}. Correspondences are visualized using both geometry and texture transfer. \label{fig:NnInit}}
\end{figure}

% Motivate: no landmarks needed, just geometry
% Introduce initialization from closest point map after standard Conf(3) registration
\subsection{Landmark-Free Correspondences}
\label{sec:LandmarkFreeCorrespondences}

Many correspondence algorithms require landmarks or other user-specified constraints to guide the optimization. Our method can compute high-quality correspondences even in the absence of landmarks, relying only on the extrinsic geometry of the input surfaces.
Given two surfaces \(A\) and \(B\), we first scale and rigidly align them. Then we compute the nearest-neighbor map which assigns to each vertex \(v_A\in V_A\) the closest face of \(B\), and vice versa.

This raw nearest-neighbor map is often noisy and discontinuous, but our initialization procedure (\cref{sec:Initialization}) effectively regularizes it, naturally smoothing the correspondence. As shown in \cref{fig:NnInit}, this step alone can help to correct the rough geometric input, producing a more coherent initial map even before the energy minimization begins.
Optimizing the Ginzburg-Landau energy refines this map into a high-quality correspondence, as seen on the right of \cref{fig:NnInit}.

Due to its simplicity and effectiveness, we use this closest-point initialization on all examples (including those with landmark constraints) unless otherwise stated. \cref{fig:ZeroPlacement,fig:LowResolution,fig:MultiresolutionSection,fig:MeshIndependence,fig:NnInit} show landmark-free correspondences computed between pairs of nearly isometric shapes differing by large deformations, while \cref{fig:BiologicalData} shows landmark-free correspondences computed between similar but non-isometric biological shapes. Despite the lack of supervision, our method produces smooth, minimal-distortion maps that correctly untangle and align salient geometric features.

%\begin{figure}[!htpb]
%    \includegraphics{example-maps.pdf}
%    \caption{Landmark-free correspondences computed using a 4-level hierarchy. Computing the correspondence between the hands \figloc{(left)} took 10 minutes, and the correspondence between the arms \figloc{(right)} took 45 minutes.}  \label{fig:ExampleMaps}
%\end{figure}

% \begin{figure}[!b]
%     \centering
%     \includegraphics{Brain}
%     \caption{}
%     \label{fig:Brain}
% \end{figure}

% \begin{figure}[!b]
%     \centering
%     \includegraphics{jaw}
%     \caption{}
%     \label{fig:Jaw}
% \end{figure}

% \begin{figure}[!b]
%     \centering
%     \includegraphics{teeth}
%     \caption{}
%     \label{fig:Teeth}
% \end{figure}

\begin{figure}
\includegraphics{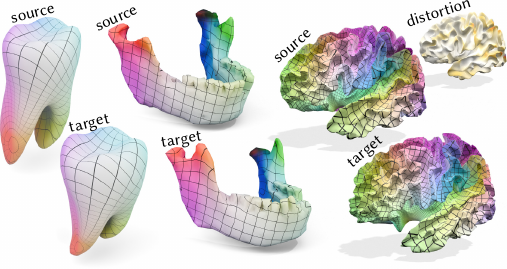}
\caption{Landmark-free correspondences computed on a variety of non-isometric biological shapes.}
\label{fig:BiologicalData}
\end{figure}

% Untangling / Robustness to Initialization
\subsection{Untangling Surface Maps}
\label{sec:UntanglingSurfaceMaps}
% As discussed in~\cref{sec:LandmarkFreeCorrespondences}, our algorithm recovers meaningful correspondences even from noisy nearest neighbor maps.
Since the Ginzburg-Landau functional is well-defined regardless of whether the section \(z\) encodes a valid bijection, our method can also repair invalid correspondences produced by other algorithms.

In particular, thanks to its structural orientation preservation, our method can correct correspondences exhibiting severe orientation reversals due to intrinsic symmetries (\cref{fig:FmapsUntangling}). We can also compute symmetric self-maps (\cref{fig:SymmetryMap}), applying an orientation-reversing constraint that effectively the identity map from the search space.

The method is also well-suited for fixing local geometric collapse: a common artifact in spectral shape matching algorithms (\eg{} functional maps), where thin structures shrink to points.
\cref{fig:FmapsCollapse} shows a typical example of a collapse from the ZoomOut algorithm~\cite{melzi2019zoomout}, appearing despite the use of initial landmarks.
In these cases, our minimization drives the zero set to expand and cover the target surface, recovering a bijection from a degenerate input.

\begin{figure}
    \includegraphics{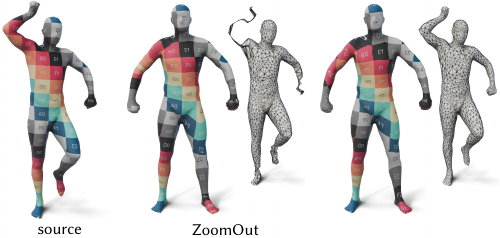}
    \caption{Spectral methods, such as Functional Maps, usually fail preserve thin shape features.
    A ZoomOut algorithm~\cite{melzi2019zoomout} initialized with landmarks produces a map that appears smooth (center). Visualizing coordinate transfer however reveals that the arms of the human shape have collapsed.
    Using this map as initialization, our algorithm (right) is able to correctly match the limbs thanks to its built-in topology and orientation preserving constraints. \label{fig:FmapsCollapse}}
\end{figure}

% Motivate: easy for us, hard for explicit representations
% Arclength param (curvilinear coordinates?) typically implies poor interior matchings
\subsection{Curve-to-Curve Correspondences}
\label{sec:CurveToCurveCorrespondences}

Like many shape correspondence algorithms, our method supports landmarks that constrain the map at isolated points (\cref{fig:Landmarks}).
But many applications require more flexible landmarks, like mapping a curve on surface \(A\) to a corresponding curve on surface \(B\) without fixing the exact pointwise map.
These constraints arise naturally when matching surfaces with sharp feature curves like boundaries or creases (\cref{fig:Boundary,fig:CurvePinning}). And they often appear in geometric morphometrics, under the name ``semilandmarks'', as curves are easier to identify than specific points~\cite{gunz2013semilandmarks}.

A na\"ive, albeit common, approach is to sample points along each curve and enforce explicit point-to-point correspondences.
However, this arbitrarily fixes the parameterization between the curves; even the ``natural'' choice of arc-length parameterization often twists the map, leading to poor interior correspondences (\cref{fig:PointCurveConstraints}).
Our implicit framework, on the other hand, handles such curve-to-curve correspondences with ease---the singularity pinning potential attracts the zeros (and therefore the mapping surface) to the two-dimensional patch \(\gamma_A \times \gamma_B \subset A \times B\) traced out by the pair of corresponding curves, without enforcing any preferential parameterization.
The correspondence between the curves then emerges according to what is energetically favorable. Importantly, the two input curves do not need to have the same number of sample points.
Finally, we note that extending this approach to mixed point-to-curve constraints is straightforward.

\begin{figure}[!htpb]
    \includegraphics{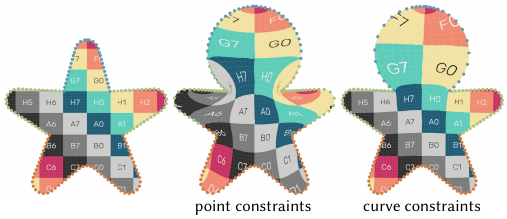}
    \caption{Enforcing curve-to-curve correspondences can be challenging, and many algorithms rely on constraining the correspondences of each point on the curve. Matching curves based on curvilinear coordinates, however, may introduce significant distortion (middle). Our framework enables the easy enforcement of curve-to-curve correspondences, allowing points to slide along the specified curves (right). This approach substantially reduces overall distortion by accurately matching the concave corners of the boundary. \label{fig:PointCurveConstraints}}
\end{figure}

% Comparisons
% Metrics: conformal distortion, isometric distortion (symmetric Dirichlet energy), area distortion |det dφ|, area of graph (this notation...)
\subsection{Comparisons with Prior Work}
\label{sec:ComparisonsWithPriorWork}

We evaluate our method against state-of-the-art approaches, checking robustness to initialization, and quality of the resulting correspondences. In particular, we compare to the constant curvature metric based inter-surface mapping (ISM) approach from~\cite{schmidt2020inter}, the adaptive triangulations (AT) method from~\cite{schmidt2023surface}, along with the reversible harmonic maps (RHM) approach of~\cite{Ezuz:2019:RHM} and hyperbolic orbifold Tutte embeddings (HOTE)~\cite{aigerman2016hyperbolic}. We additionally compare and highlight important differences of our approach with functional maps based correspondences.

\paragraph{Robustness to Initialization}
A key advantage of our approach lies in its ability to recover high-quality maps even from poor initial correspondences (\cref{sec:LandmarkFreeCorrespondences,sec:UntanglingSurfaceMaps}). In \cref{fig:ComparisonInitialization}, we replicate the initialization experiment from~\cite[Fig. 13]{schmidt2020inter}, computing correspondence from increasingly distorted initial states.
RHM, which directly minimizes the Dirichlet energy, consistently becomes trapped in local minima and fails to realign thin features or correct large distortions.
While ISM performs significantly better, it achieves consistent results only for the first three initializations before diverging.
In contrast, our method converges to a nearly identical, low-distortion map across all initializations, verifying its untangling capability discussed in \cref{sec:UntanglingSurfaceMaps}.

\begin{figure}[!tpb]
    \includegraphics{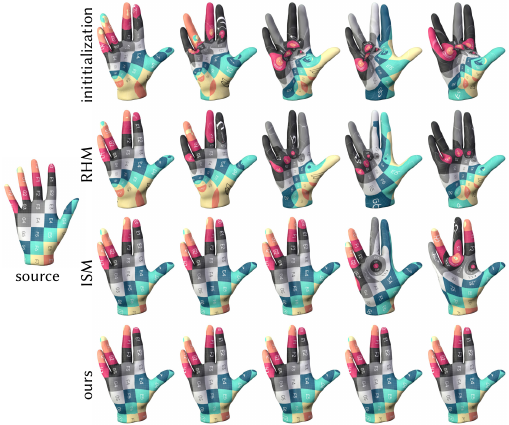}
    \caption{We replicate the experiment from~\citet{schmidt2020inter} to evaluate robustness against poor initializations (top row). The RHM method~\cite{Ezuz:2019:RHM} becomes trapped in local minima of distortion for each initialization. While \citet{schmidt2020inter} (ISM) performs better and achieves consistent results for the first three initializations, our method produces nearly identical low-distortion mappings across all initializations.
    \label{fig:ComparisonInitialization}}
\end{figure}

\paragraph{Map Quality and Distortion}
We further compare the quality of the final maps against ISM, which also encodes the overlay mesh induced by the correspondences.
As shown in \cref{fig:ComparisonInterSurfaceMaps}, while both algorithms converge to geometrically similar correspondences, our implicit Ginzburg-Landau minimization results in consistently lower distortion energies (measured by both the area of the correspondence graph and symmetric Dirichlet energy).
The difference is particularly noticeable in the smoothness of the distortion and the regularity of the texture transfer. Our method avoids the localized distortion spikes often seen in explicit remeshing-based approaches.
\begin{figure}
    \includegraphics{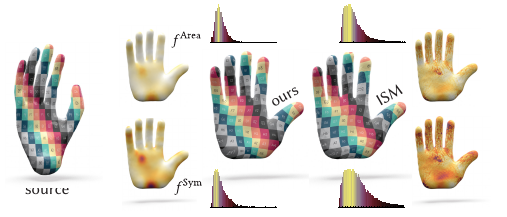}
    \caption{Both our algorithm and the algorithm of~\citet{schmidt2020inter} (ISM) represent the common subdivision of the overlay mesh induced by the correspondence. It is needed to evaluate distortion energies of the correspondence intrinsically. While both algorithms find similar correspondence, implicit area minimization produces a less distorted map, even at the finer scales. This is reflected in the smoother distortion distributions, measured both by the area of the correspondence graph given by integrating \(f^{\text{Area}}(\sigma_1,\sigma_2)=\smash{({(1+\sigma_1^2)(1+\sigma_2^2)}})^{1/2}\), and the symmetric Dirichlet energy, integrating \(f^{\text{Sym}}(\sigma_1,\sigma_2)=\smash{\sigma_1^2+\sigma_2^2+\sigma_1^{-2} + \sigma_2^{-2}}\). The histogram is colored by the height of the bars.
    \label{fig:ComparisonInterSurfaceMaps}}
\end{figure}

\begin{figure}[!tpb]
    \includegraphics{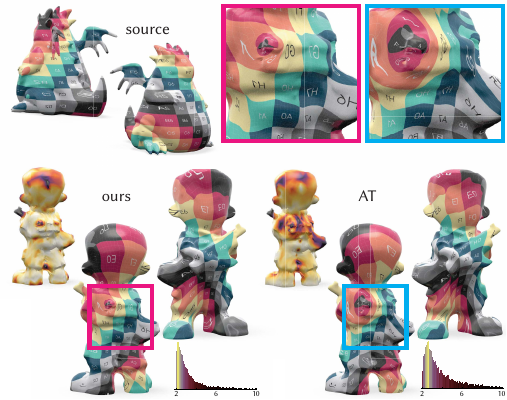}
    \caption{Distortion is inevitable between non-isometric shapes. Compared with the adaptive triangulations (AT) method~\cite{schmidt2023surface}, our method produces dramatically less distorted correspondences (as measured by the symmetric Dirichlet energy) that better respect the intrinsic symmetries of the model. \label{fig:NonIsometric}}
\end{figure}

\begin{figure}[!htpb]
    \includegraphics{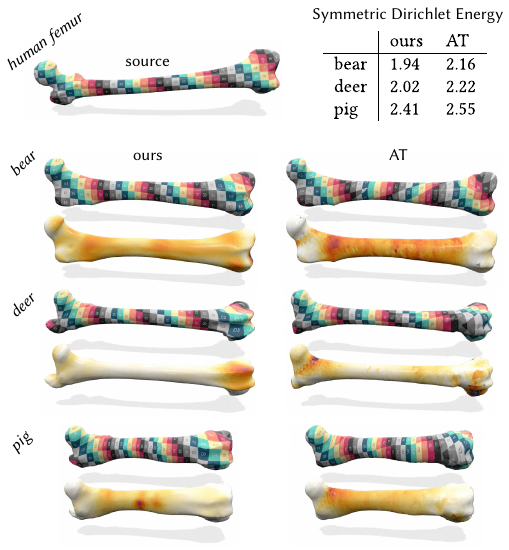}
    \caption{We use our matching algorithm to compare a human femur with the femurs of a bear, deer, and pig, each computed using four landmarks. Our method (\figloc{left}) consistently achieves lower distortion than the algorithm of~\citet{schmidt2023surface} (AT) across all three pairs. While both methods realize the same similarity ordering, the correspondence computed using AT exhibits spurious twisting artifacts that are absent in our results.
    \label{fig:FemursComparison}}
\end{figure}

We also compare the distortion to AT on a challenging pair of inputs with significant non-isometric deformation. \cref{fig:NonIsometric} shows that while AT produces a valid bijective map, our method produces a more aligned correspondence with lower symmetric Dirichlet energy. While not explicitly enforced, our method better respects the intrinsic symmetries of the shape (\eg, preserving the left-right symmetry of the human model). In the low-distortion, but non-isometric, regime important in applications like bone registration (\cref{fig:FemursComparison}) our algorithm consistently produces lower distortion correspondences that better respect the intrinsic shape symmetries.

Lastly, while our method and HOTE both produce similar correspondences away from the landmarks, HOTE introduces extremely high distortion near these distinguished points (\cref{fig:OrbifoldComparison}). By contrast, minimizing our global distortion energy produces a correspondence which extends smoothly over landmark points.

\begin{figure}[!htpb]
    \includegraphics{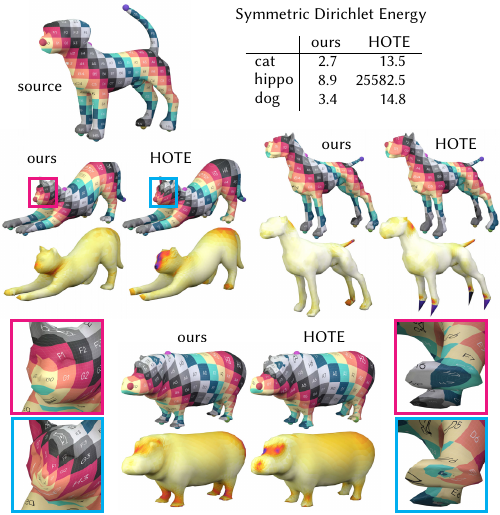}
    \caption{Our method (\figloc{left}) consistently yields lower distortion than maps produced by the hyperbolic orbifold Tutte embedding of~\citet{aigerman2016hyperbolic} (HOTE) across all three pairs. In contrast, HOTE does not minimize global mapping distortion and introduces severe artifacts near landmarks. Distortion is measured between the source mesh and the mesh obtained by mapping its vertices onto the target shape.
    \label{fig:OrbifoldComparison}}
\end{figure}

%\begin{figure}[!htpb]
%    \caption{Our algorithm can be applied to tasks such as spherical parameterization. The resulting mapping is bijective and exhibits low distortion, even when the sphere sampling is extremely coarse. Compared to other spherical parameterization methods---(a) mean curvature flow~\cite{kazhdan2012can}, (b) spherical Tutte embedding~\cite{gotsman2003fundamentals,friedel2007unconstrained}, and (c) Willmore flow~\cite{Crane:2013:RFC}---our approach yields lower isometric distortion.\label{fig:SphericalParametrization}}
%\end{figure}

%\subsection{Performance}
%\label{sec:Performance}

%\begin{figure}[!htpb]
%    \caption{The maps computed by our algorithm are constrained to preserve normal orientation of the based surface. Intrinsic algorithms (such as functional maps) often struggle to maintain orientation preservation. We initialize our algorithm with three maps obtained from the functional map representation: (a) the direct map, (b) a linear combination of the direct and inverse map and (c) the inverse map. The initial step already output a mapping with correct orientation and minimizing the GL functional lead to further distortion minimization.\label{fig:OrienationPreservation}}
%\end{figure}

\subsection{Implementation and Parameters}

We used L-BFGS to minimize the Ginzburg-Landau energy, terminating when the norm of the projected gradient fell below \(10^{-5}\) or after a maximum of 1000 iterations.
Comprehensive pseudocode can be found in the supplementary material, and C++\footnote{\url{https://github.com/yousufmsoliman/implicit-minimal-surfaces}}, MATLAB\footnote{\url{https://github.com/etcorman/implicit-minimal-surfaces}}, and Python\footnote{\url{https://github.com/RobinMagnet/implicit-minimal-surfaces}} implementations of our method are available.

The main parameter of our method is the Ginzburg-Landau parameter \(\lambda\), weighting the circular well potential. As discussed in \cref{sec:DiscreteGinzburgLandauMinimization,sec:TheGinzburgLandauParameter}, setting \(\lambda\approx 100\lambda_{A\times B}\) generally produces high-quality results.
However, we observed that examples requiring significant untangling or large deformations from the initialization benefit from a simple annealing scheme, using a smaller \(\lambda\) first to align global features before refining with the default value.
When using landmarks, we set equal widths for the Gaussian kernels, \(\sigma_A = \sigma_B\), setting them depending on the landmarks distribution.
We acknowledge that this parameter is sensitive to the specific configuration: the optimal \(\sigma\) depends both on the geometric scale and on the distribution of the landmarks. While the precise value of this parameter does not dramatically change the correspondences, we select this parameter manually to ensure the potential wells are sufficiently wide to attract the zero set.
We expect that future work on adaptive pinning potentials will allow this parameter to be determined automatically.
We refer the reader to the supplementary material for a complete list of parameters used in each figure.

\section{Limitations and Future Work}
\label{sec:LimitationsAndFutureWork}
While our algorithm offers new perspectives and significant advantages in terms of correspondence quality relative to prior work, it also faces some challenges and suggests avenues for future research.

\paragraph{Performance}
Without a multiresolution hierarchy, the runtime to compute an implicit minimal surface scales slightly superlinearly with the product of the vertex counts (\cref{fig:Performance}). On a typical pair of models with \(|V_A|\approx |V_B|\approx 500\), our algorithm takes approximately \(20\) seconds, while for a pair of models with \(|V_A|\approx |V_B|\approx 5000\) our algorithm takes approximately \(1\) hour (implemented in C++, measured on an Intel i7-14700K CPU with 64GB of RAM).
While the direct application of our method is significantly slower than remeshing based approaches such as~\citep{schmidt2023surface}, it is competitive with ISM approach of~\citep{schmidt2020inter}, which reports a runtime of about \(3\) hours when \(|V_A|\approx |V_B|\approx 4000\).
\begin{figure}[!tpb]
    \includegraphics{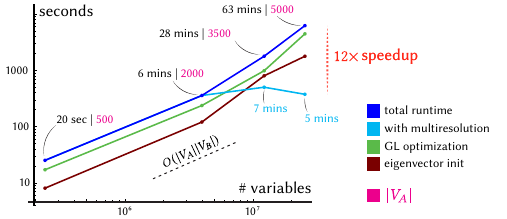}
    \caption{Computing implicit minimal surfaces scales roughly linearly in the number of variables, with the eigenvalue initialization taking a shade less than half the total runtime. By adopting a multiresolution acceleration scheme, we decrease the effective runtime significantly.\label{fig:Performance}}
\end{figure}

The runtime of our algorithm is dominated by the eigenvalue initialization and the minimization of the Ginzburg-Landau functional; the remaining operations, including the construction of surface connections and overlay mesh extraction, are negligible. For the former, we employ a na{\"i}ve LOBPCG solver, which could be significantly accelerated by incorporating algebraic or geometric multigrid preconditioners~\cite{xu2017algebraic} adapted to the tensor-product structure of the operator. Similarly, for the energy minimization, we rely on a generic off-the-shelf algorithm (L-BFGS). Using a custom solver that uses the structure of the Ginzburg-Landau energy to exploit parallelization would likely provide significant speedups. These costs are partially mitigated by our multiresolution hierarchy (\cref{sec:MultiresolutionHierarchy}), which shifts most of the computation to coarser meshes. For the example in~\cref{fig:MultiresolutionSection} (\(|V_A|\approx |V_B|\approx 5000\)) we obtained a \(12\times\) speedup and computed a distortion minimizing bijective correspondence in only \(5\) minutes.

Memory limitations place a hard constraint on the input mesh size: although the mass and stiffness matrices are never assembled in the product space, the section $z$ requires $|V_A| \times |V_B|$ storage, exceeding capacity for pairs with greater than \(\approx\)70k vertices each. By using a decimated proxy mesh to compute the correspondence there are further opportunities to decouple the runtime and memory requirements from the input resolution.

% The connection Laplacian and mass matrices factor across the two meshes and need not be assembled in the product space. However, memory limitations present a primary challenge: the section $z$ requires $|V_A| \times |V_B|$ storage---exceeding memory capacity for pairs with more than roughly 70k vertices each. These costs are partially mitigated by multiresolution hierarchies, which shift most of the computation to coarser meshes (the solution on a pair of 5k meshes in~\cref{fig:MultiresolutionSection} was computed in \(\approx 5\) minutes). Despite the simplicity of the Ginzburg-Landau functional, we have not leveraged its structure to develop a custom, potentially faster or more memory-efficient optimizer that could exploit parallelization or handle a very large number of variables more effectively. For simplicity, we used an off-the-shelf algorithm (L-BFGS) which is not optimized for very large scale problems.
\begin{figure}[!htpb]
    \includegraphics{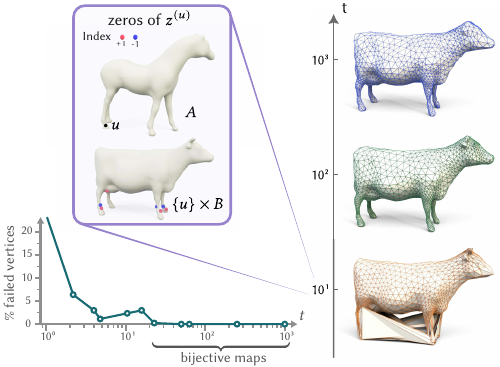}
    \caption{Coordinate transfer from the horse $A$ into the cow $B$, computed for three values of \(\lambda = t\,\lambda_0\) (right). A large value of \(\lambda\) produces clean correspondences, but the map degrades near regions of high curvature when \(\lambda\) decreases.
    The graph \emph{(bottom left)} plots the average percentage of vertices whose slices have more than one zero, thus breaking bijectivity. We obtain bijective maps once $t \gtrsim 50$.    Inspecting the failure at $t=10^1$, we find that for a vertex $u\in A$ mapped incorrectly, the slice $z^{(u)}$ on $B$ has several zeros rather than one. Their indices (red $+1$, blue $-1$) sum to $1$, as required by our construction, and one of the $+1$ zeros lie on the correct location.
    This suggests that a specific zero-selection heuristic could recover valid correspondences even in these degenerate cases.
    \label{fig:FailureMode}}
\end{figure}

\paragraph{Bijectivity}
As mentioned in \cref{sec:DiscreteGinzburgLandauMinimization}, the choice of Ginzburg-Landau parameter can have a big effect on the quality of the output map. \cref{fig:FailureMode} shows maps computed between the same pair of surfaces for a variety of parameters \(\lambda = t \lambda_0\). In general, we observe that when \(\lambda\) is not sufficiently large, the map degrades in regions of high curvature, yielding a non-bijective correspondence between the two surfaces. In practice, we find that increasing the Ginzburg-Landau parameter allows us to compute low-distortion bijective maps, but further theoretical work would be required to provide guarantees on the quality or bijectivity of the output maps.

First of all non-convexity of the Ginzburg-Landau functional means that the optimization may converge to local minima. While theoretical work suggests that critical points---beyond just global minima---can describe minimal surfaces, it remains unclear whether \emph{all} critical points correspond to such surfaces. Moreover, more analysis is needed to characterize (and constrain) when these implicit minimal surfaces are graphs over both factors (\ie{} bijections).

\paragraph{Generalization}
Extending the algorithm to general surfaces, not limited to genus zero surfaces with boundaries, is a compelling direction for future work. The topology of the implicit surfaces, encoded in the curvature of the connection, needs to be modified so that it represents surfaces that are the graphs of bijective correspondences. In addition to total curvature \(2\pi\) on each slice, additional curvature must be concentrated, according to the map homotopy type, on the closed two-dimensional surfaces corresponding to the product of homology generators. Beyond changing the topology of the implicit surfaces, additional complications arise when minimizing the Ginzburg-Landau functional that may obstruct extracting a correspondence. For instance, the energy density of critical points of the Ginzburg-Landau functional may not concentrate around a singularity as \(\varepsilon\to 0\), instead spreading out over the entire space according to a harmonic one-form.

In principle, our implicit representation can also be generalized to codimension-3 manifolds, replacing complex line bundles with rank-3 vector bundles. Looking for implicit codimension-3 minimal surfaces in the product of two volumetric domains would be the natural generalization of our method to volumetric correspondences.

\begin{wrapfigure}{r}{125pt}
    \includegraphics{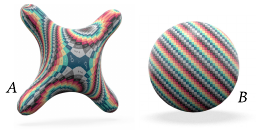}
\end{wrapfigure}
Finally, we mention that beyond bijective correspondences, our method can compute \emph{covering maps} instead of bijections. For instance, by picking a connection on \(A\) which has total curvature \(4\pi\) instead of \(2\pi\), we can compute a double covering map where every point of \(B\) is covered by two points of \(A\).

\begin{acks}
Thanks to Aria Halavati for teaching us about the complex line bundle approach to codimension-two minimal surfaces.
\end{acks}

% Bibliography
\bibliographystyle{ACM-Reference-Format}
\bibliography{minimal_ref}

\appendix

\section{Riemannian Geometry in the Product of Surfaces}
\label{sec:RiemannianGeometryInTheProductOfSurfaces}
Consider a pair of oriented smooth surfaces \(A\) and \(B\) with Riemannian metrics \(g_{A}\) and \(g_{B}\), respectively. The product space $A \times B$ is a Riemannian 4-manifold endowed with the metric $g_{A \times B}(X_A \oplus X_B, Y_A \oplus Y_B) := g_A(X_A,Y_A) + g_B(X_B,Y_B)$.
\begin{proposition}
    Given \(\varphi:A\to B\) the area of its graph \(\Sigma_{\varphi}\) can be expressed as the integral of the singular values \(\sigma_1,\sigma_2\) of the deformation \(d\varphi\):
    \[\textup{Area}(\Sigma_{\varphi}) = \int_{A}\sqrt{(1+\sigma_1^2)(1+\sigma_2^2)}~\vol_{A}.\]
    \label{prp:AreaIntegral}
\end{proposition}
\begin{proof}
    The graph can be parameterized by $F_{\varphi} : A \rightarrow \Sigma_{\varphi}$,  $F_{\varphi}(p) := (p,\varphi(p))$, and so
    $\textup{Area}(\Sigma_{\varphi})$ can be computed by integrating the area form induced by the parameterization:
    \begin{align*}
        \textup{Area}(\Sigma_{\varphi}) &= \int_{A} \det\left( F_{\varphi}^\star g_{A \times B} \right)^{\tfrac{1}{2}} \vol_{A} \\
        &= \int_{A} \det\left( I + g_A^{-1} \varphi^\star g_B \right)^{\tfrac{1}{2}} \vol_{A} \\
        &= \int_{A} \left[ 1 + \det\left( g_A^{-1} \varphi^\star g_B \right) + \tr\left( g_A^{-1} \varphi^\star g_B \right)\right]^{\tfrac{1}{2}} \vol_{A} \\
        &= \int_{A} \left[ 1 + \sigma_1^2 \sigma_2^2 + \sigma_1^2 + \sigma_2^2\right]^{\tfrac{1}{2}} \vol_{A} .
    \end{align*}
\end{proof}

As expected the change for variables to $B$ leads to the same function applied to the inverse singular values:
\begin{align*}
   \textup{Area}(\Sigma_{\varphi}) &= \int_{B} \left[ 1 + \sigma_1^2 \sigma_2^2 + \sigma_1^2 + \sigma_2^2 \right]^{\tfrac{1}{2}} \frac{\det\left( (\varphi^{-1})^\star g_B \right)^{\tfrac{1}{2}}}{\det\left( g_B \right)^{\tfrac{1}{2}}} \vol_B \\
      &= \int_{B} \big[ 1 + \sigma_1^2 \sigma_2^2 + \sigma_1^2 + \sigma_2^2 \big]^{\tfrac{1}{2}} \frac{1}{\sigma_1 \sigma_2} \vol_B \\
      &= \int_{B} \Big[ \frac{1}{\sigma_1^2 \sigma_2^2} + 1 + \frac{1}{\sigma_2^2} + \frac{1}{\sigma_1^2} \Big]^{\tfrac{1}{2}} \vol_B .
\end{align*}

\section{The Product-Space Dirichlet Energy}
\label{sec:TheProductSpaceDirichletEnergy}
Here we discretize the smooth expression of Dirichlet energy in \cref{eq:DirichletEnergyDecomposition} to obtain a discrete Dirichlet energy on the product space.
We consider a connection \(\nabla\) that is the tensor product of connections \(\nabla^A\) and \(\nabla^B\) on \(A\) and \(B\), respectively.
For convenience, we reproduce \cref{eq:DirichletEnergyDecomposition} below:
\begin{align*}
    \DirichletEnergy(z) = \int_{A}\left(\int_{B}|\nabla z|_{B}^2~\vol_{B}\right)\vol_{A} + \int_{B}\left(\int_{A}|\nabla z|_{A}^2~\vol_{A}\right)\vol_{B}
\end{align*}
Discretizing the second term yields:
\begin{align*}
    & \int_{A \times B} |\nabla^A z(p,q)|^2 \vol_{A \times B}(p,q) \\
    & = \sum_{\substack{\ijk \in F_A \\ abc \in F_B}} \underbrace{\int_{\ijk}\int_{abc}\Big| \sum_{\substack{u \in \ijk \\ v \in abc}} Z_{u,v} \nabla^A \phi^A_u(p) \phi^B_v(q) \Big|^2 \vol_{A}(p)\vol_{B}(q)}_{X_{\ijk,abc}}.
\end{align*}
The integrand evaluates to
\begin{align*}
    X_{\ijk,abc} &= \sum_{\substack{m \in \ijk \\ n \in abc}} \bar{Z}_{m,n} \Bigg[\sum_{\substack{u \in \ijk \\ v \in abc}} \left( \int_{\ijk} (\nabla^A \phi^A_m)^\star(p) \nabla^A \phi^A_u(p) \vol_{A}(p)\right) \cdot \\ & \quad \hspace{7.5em}\left(Z_{u,v} \int_{abc} \bar{\phi}^B_n(q) \phi^B_v(q) \vol_{B}(q) \right)\Bigg],
\end{align*}
where we recognize the integrals as the components of connection Laplacian of $A$ and the mass matrix of $B$. Hence, the connection Laplacian on $A\times B$ has a tensor-product structure: \begin{equation}
          L_{A\times B}^{\nabla} = L_A^\nabla \otimes M_B^\nabla + M_A^\nabla \otimes L_B^\nabla
  . \end{equation}

Using a matrix representation $Z \in \mathbb{C}^{V_A \times V_B}$ of the complex field \(z: A \times B \to \mathbb{C}\), the discretization of Equation~\ref{eq:DirichletEnergyDecomposition} only depends on left and right multiplications of the connection Laplacian and mass matrices defined on each manifold:
\begin{equation*}
    \DirichletEnergy(Z) = \frac{1}{2} \langle Z, L_A^\nabla Z (M_B^\nabla)^\top \rangle_\CC + \frac{1}{2} \langle Z, M_A^\nabla Z (L_B^\nabla)^\top \rangle_\CC.
\end{equation*}
\section{The Ginzburg-Landau Parameter}
\label{sec:TheGinzburgLandauParameter}
For the following result, we consider a Hermitian positive-definite matrix $A\in\CC^{d\times d}$ along with a diagonal matrix $M\in \RR^{d\times d}$---for our application, they will be the connection Laplacian and mass matrix on a four-dimensional cell complex, respectively.
\begin{lemma}
    Let $\lambda_{0}>0$ be the smallest eigenvalue of the generalized eigenvalue problem $Az = \lambda M z$ . Consider the discrete energy
    \[
        e(z) = \tfrac{1}{2}\langle Az, z\rangle + \tfrac{\lambda}{4}\langle M u,u\rangle,
    \]
    where $u_i = 1-|z_i|^2$ for every vertex $i=1,\dots,d$. If $0 < \lambda \leq \lambda_0$ then the only critical point of $e$ is given by $z \equiv 0$.
    \label{lem:global_mini_GL}
\end{lemma}
\begin{proof}
    Taking the scalar product between a stationary point $\bar{z}$ and the gradient $\nabla e(\bar{z})$, we obtain:
    \begin{equation}
        \langle\bar{z},\nabla e(\bar{z})\rangle = \langle\bar{z}, A \bar{z}\rangle - \lambda\langle\bar{z}, M\bar{z}\rangle + \lambda \langle |\bar{z}|^2, M |\bar{z}|^2 \rangle = 0.
            %+ \lambda\sum_{i=1}^d M_{ii}|\bar{z}_i|^4 = 0.
    \end{equation}
    Using the fact that the smallest eigenvalue of $A$ is greater than $\lambda$, we conclude that $\bar{z}$ must be equal to zeros:
    \begin{equation}
        0 \leq \langle\bar{z}, A \bar{z} -\lambda M\bar{z}\rangle = - \lambda \langle |\bar{z}|^2, M |\bar{z}|^2 \rangle = 0 \leq 0 % \lambda\sum_{i=1}^d M_{ii}|\bar{z}_i|^4 \leq 0
    \end{equation}
\end{proof}
% If the mass matrix is not-diagonal, the above result still holds if we take \(\lambda_0\) to be the smallest eigenvalue of the generalized eigenvalue problem $Az=\lambda \hat{M}z$, where $\hat{M}$ is the result of lumping $M$ to the diagonal, \ie, $\hat{M}_{ii} = \sum_{j}M_{\ij}$ for $i=1,\dots,d$.
In the continuous setting, the connection Laplacian $\Delta^\nabla=\Delta_A^\nabla \otimes I_B + I_A \otimes \Delta_B^\nabla$ on $A\times B$ decomposes as a sum of connection Laplacians on each shape.
Its minimal eigenvalue is therefore the sum of the minimal eigenvalues on each shape: $\lambda_0 = \lambda_A^\nabla + \lambda_B^\nabla$, where $\lambda_A^\nabla$ and $\lambda_B^\nabla$ are the minimal eigenvalues of $\Delta_A^\nabla$ and $\Delta_B^\nabla$ respectively (defined with respect to the corresponding volume form on each shape).

In practice, however, due to differing discretization choices, we use $L_{A\times B}^{\nabla} = L_A^\nabla \otimes M_B^\nabla + M_A^\nabla \otimes L_B^\nabla$ for the Dirichlet energy, but $M = M_A \otimes M_B$ for the potential energy, using lumped mass matrices $M_A$ and $M_B$.
The generalized eigenvalue problem $Az = \lambda M z$ then doesn't strictly separate.
Motivated by the continuous setting, we approximate the stability threshold using $\lambda = \lambda_A^\nabla + \lambda_B^\nabla$ where $\lambda_A^\nabla$ and $\lambda_B^\nabla$ are the smallest eigenvalues of $L_A^\nabla$ with respect to $M_A^\nabla$, and $L_B^\nabla$ with respect to $M_B^\nabla$, respectively.
As both mass matrices approximate the same continuous volume form, we expect this value to provide a close approximation of the true stability threshold.

\section{Alternative Surface Connections}
\label{app:DiscreteSurfaceConnections}
We describe two discrete connections with curvature \(\Omega = \tfrac12\Omega^{LC}\): one based on concentrating and redistributing the curvature of the Levi-Civita connection and another based on spin structures. Below, \(M\) is a triangle mesh describing one of the two input surfaces, with $\connectionUnitary^{LC}$ and $\Omega^{LC}$ as in \cref{sec:DiscreteLeviCivitaConnection}.

While we obtained similar correspondences irrespective of the choice of connection, the available multiresolution schemes differ (\cref{sec:MultiresolutionHierarchy}). For instance, to prolongate the sections of the bundles constructed below requires keeping track of parallel transport maps across the hierarchy of meshes. Further refinements and generalizations of our approach may be possible by changing the surface connections and their curvature.

\paragraph{A Vector Field Connection}
We can modify the discrete Levi-Civita connection by computing an ``offset connection'' $\tilde{\connectionUnitary}$ so that $\connectionUnitary = \tilde{\connectionUnitary}\ \connectionUnitary^{LC}$ is compatible with $\Omega$.
Similar to our discretization of the skyscraper bundle, we first select an arbitrary face $f_0$, and define \[\tilde{\Omega} = \Omega^{LC} - 2\pi (\chi(M)-1)\delta_{f_0},\] where $\delta_{f_0}$ is the Kronecker delta on faces. By construction, $\tilde{\Omega}$ sums to $2\pi$ on $M$. Due to the \(2\pi\) ambiguity in the curvature of a discrete complex line bundle, it is also compatible with $\connectionUnitary^{LC}$. The offset connection can now be computed by solving a Poisson equation: we set $\tilde{\connectionUnitary}_{\ij} = e^{\imath\tilde{\connectionAngle}_{\ij}}$ where \(\tilde{\connectionAngle}\) solves
\begin{equation}
    \argmin_{\tilde{\connectionAngle}} \|\tilde{\connectionAngle}\|^2  \quad \text{s.t.} \quad \mathsf{d}_1 \tilde{\connectionAngle} = \Omega - \tilde{\Omega}.
\end{equation}

\paragraph{Spin Connections}
A canonical choice on genus zero surfaces is given by a discrete spin connection~\cite{Chern:2018:SFM}, which is given by computing a square-root of the Levi-Civita connection \[\connectionUnitary^{\text{spin}}_{\ij} = \pm (\connectionUnitary^{LC}_{\ij})^{1/2}.\] Unlike the other constructions we discussed, the construction does not depend on the choice of an arbitrary face. The signs of the square-root, however, must be chosen appropriately so that \(\connectionUnitary^{\text{spin}}\) is compatible with \(\Omega\). A simple spanning tree based algorithm to choosing the signs is given in~\cite[Algorithm 3]{Chern:2018:SFM}.

\section{Edge-Edge Intersections}
\label{sec:EdgeEdgeIntersections}

\begin{wrapfigure}{r}{72pt}
\includegraphics{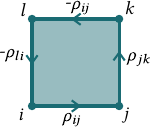}
\end{wrapfigure}
We identify the intersection between an edge of \(A\) and an edge of \(B\) by locating zeros in ``edge-edge'' faces of \(A \times B\). On such a face complex section values are interpolated using the tensor product of the basis functions along each edge, so a point \((s, t)\) in the face is a zero if and only if
\begin{equation*}
\begin{aligned}
0 &= (1-s) (1-t) e^{\imath (s \rho_\ij + t \rho_{\jk})} z_i + s (1-t) e^{\imath (-(1-s) \rho_\ij + t \rho_{\jk})} z_j\\
&\quad + s t e^{\imath (-(1-s)\rho_\ij -(1-t)\rho_\jk)} z_k + (1-s) t e^{\imath (s\rho_\ij - (1-t)\rho_\jk)} z_l.
\end{aligned}
\end{equation*}
Dividing through by \(z_i e^{\imath (s\rho_\ij + t\rho_\jk)}\), we obtain
%\begin{equation}
%\begin{aligned}
%0 &= (1-s) (1-t) + s (1-t) e^{-\imath\rho_\ij} \frac{z_j}{z_i}\\
%&\quad+ s t e^{-\imath (\rho_\ij + \rho_\jk)} \frac{z_k}{z_i} + (1-s) t e^{-\imath \rho_\jk} \frac{z_l}{z_i}.
%\end{aligned}
%\end{equation}
\begin{equation*}
\begin{aligned}
0 &= (1-s) (1-t) + s (1-t) \frac{z_j}{r_\ij z_i} + s t \frac{z_k}{r_\ij r_\jk z_i} + (1-s) t \frac{z_l}{r_\jk z_i}.
\end{aligned}
\end{equation*}
Now \(t\) is the solution to the following quadratic equation and \(s\) is the associated quotient:
% \begin{equation}
% t^2 \Im(a\bar{c}) + t \Im(a\bar{d} - c\bar{b}) + \Im(b\bar{d}) = 0 \quad s =-\frac{\Re(ct+d)}{\Re(at+b)},
% \end{equation}
\begin{equation*}
t^2 \operatorname{Im}(a\bar{c}) + t \operatorname{Im}(a\bar{d} - c\bar{b}) + \operatorname{Im}(b\bar{d}) = 0 \quad s =-\frac{\operatorname{Re}(ct+d)}{\operatorname{Re}(at+b)},
\end{equation*}
where
% \begin{align}
% & b = -1 + e^{-\imath\rho_\ij} \frac{z_j}{z_i}, \, c = -1 + e^{-\imath \rho_\jk} \frac{z_l}{z_i}, \, d = 1 \\
% & a = -1 - b - c + e^{-\imath (\rho_\ij + \rho_\jk)} \frac{z_k}{z_i}.
% \end{align}
\begin{align*}
& b = -1 + r_\jk^{-1} \frac{z_j}{z_i}, \, c = -1 + r_\jk^{-1} \frac{z_l}{z_i}, \, d = 1 \\
& a = -1 - b - c + r_\jk^{-1} r_\jk^{-1} \frac{z_k}{z_i}.
\end{align*}

% \end{document} %========= cut off supplemental material

\clearpage
\setcounter{page}{1} % restart page numbering
\setcounter{section}{0} % restart section numbering

\noindent%
{\Huge Supplemental Material}

\hspace{.5em}To \emph{Implicit Minimal Surfaces for Bijective Correspondences} by Etienne Corman, Yousuf Soliman, Robin Magnet, and Mark Gillespie

\section{Pseudocode}
\label{sec:Pseudocode}

This supplement provides detailed pseudocode for computing bijective correspondences via minimal surfaces.

Subroutines and quantities not defined in pseudocode are described in the list below.
\begin{itemize}[leftmargin=*]
    \item \(a_{\ijk}\) --- area of face \(\ijk\).
    \item \(\theta_i^{\jk}\) --- corner angle for vertex \(i\) in face \(\ijk\).
    \item \(\mathsf{d}_1 \in \RR^{F \times E}\) --- discrete exterior derivative \cite{Desbrun:2005:DEC}.
    \item \(*_1 \in \RR^{E \times E}\) --- Hodge star \cite{Desbrun:2005:DEC}.
    \item \(M_S \in \RR^{V \times V}\) --- the (scalar) vertex lumped mass matrix on triangle mesh \(S\). Note that this is different from the complex connection mass matrix \(M^\nabla_S\) computed in \cref{alg:BuildFEMConnectionMatrices}.
    \item \(d_S(p, q), d_S(p, \gamma)\) --- the geodesic distance along \(S\) from point \(p \in S\) to point \(q \in S\), or to curve \(\gamma \subset S\). (See \eg{} \cite{Crane:2017:HMD}.)
    \item \Proc{LinearSolve}(\(\Asf,\bsf\)) --- solves a sparse linear system \(\Asf\xsf=\bsf\).
    \item \Proc{MinEigenvalue}(\(\Asf,\Bsf\)) --- computes the smallest eigenvalue of the pair \(\Asf, \Bsf\), \ie{} the smallest \(\lambda\) so that exist an \(\xsf\) with \(\Asf \xsf = \lambda \Bsf \xsf\).
    \item \Proc{MinEigenvector}(\(\xsf \mapsto \Asf \xsf, \xsf \mapsto \Bsf\xsf\)) --- computes the eigenvector of the pair \(\Asf, \Bsf\) with smallest eigenvalue. For efficiency, we formulate the problem via the operators \(\Asf, \Bsf\) and avoid assembling the whole product space matrices.
    \item \Proc{LBFGS}(\((f, \nabla f), x_0\)) --- minimize \(f\) using LBFGS, starting from initial point \(x_0\), returning optimized point \(x\)
    \item \Proc{NormalizeToUnitSurfaceArea}(\(M\)) --- scale triangle mesh \(M\) so that it has surface area 1.
\end{itemize}

\paragraph{Product-space matrices} Recall from  \cref{sec:CellComplexesAndProductMeshes} that the product mesh \(A \times B\) has vertex set \(V_{A \times B} = V_A \times V_B\), and so it is convenient to represent discrete sections \(z \in \CC^{V_{A \times B}}\) by matrices \(Z \in \CC^{V_A \times V_B}\). Similarly, the product mesh has edge set \(E_{A \times B} = E_A \times V_B \cup V_A \times E_B\), so we can represent a connection \(\connectionUnitary \in \CC^{E_{A \times B}}\) by a \emph{pair} of matrices \(\connectionUnitary^{E,V} \in \CC^{E_A \times V_B}\) and \(\connectionUnitary^{V,E} \in \CC^{V_A \times E_B}\), where \(\connectionUnitary^{E,V}_{e_A, v_B}\) gives the entry for the product-space edge \(e_A \times v_B\), and \(\connectionUnitary^{V,E}_{v_A, e_B}\) gives the entry for the product-space edge \(v_A \times e_B\). We often abuse notation and write \(r_{e_A, v_B}\) for the entries of the first matrix in \(\CC^{E_A \times V_B}\), and write \(r_{v_A, e_B}\) for the entries of the second matrix in \(\CC^{V_A \times E_B}\). Finally, the product mesh has face set \(F_{A \times B} = F_A \times V_B \cup E_A \times E_B \cup V_A \times F_B\), and thus we write vectors \(\Omega \in \RR^{F_{A \times B}}\) as triplets of matrices in \(\RR^{F_A \times V_B}\), \etc{}. We write the \(i\)th row of a matrix \(M\) as \(M_{i,\bullet}\) and the \(j\)th column as \(M_{\bullet, j}\).

\newpage

\begin{algo}{\Proc{SurfaceConnection}$(S)$}
    \label{alg:SurfaceConnection}
    \begin{algorithmic}[1]
        \InputConditions{A triangle meshes \(S = (V,E,F)\) with edge lengths \(\ell\).}
        \OutputConditions{A connection $\connectionUnitary \in \CC^{E}$ and compatible curvature $\Omega\in\RR^{F}$ such that \(\sum_f \Omega_f = 2\pi\)}
        \State $\Omega_{\ijk}\gets\tfrac12(\tilde{\theta}_i^{\jk} + \tilde{\theta}_j^{\ki} + \tilde{\theta}_k^{\ij} - \pi)$ \Comment{for every face $\ijk\in F$}
        \State Pick an arbitrary face $f_0\in F$
        \State $L\gets \dsf_1*_1^{-1}\dsf_1^\top$ \Comment{2-form Laplacian on $S$}
        \State $u\gets$\Proc{LinearSolve}$(L, \Omega - 2\pi\delta_{f_0})$
        \State $r_{\ij}\gets \exp(\imath *^{-1}_{\ij}(u_{\jil} - u_{\ijk}))$ \Comment{for every edge $\ij\in E$}
        \State \Return $r,\Omega$
    \end{algorithmic}
\end{algo}

\begin{algo}{\Proc{FindTriangleZero}$(\ijk, \fieldRotationAngle^0,\Omega^0,|z|)$}
   \label{alg:FindTriangleZero}
   \begin{algorithmic}[1]
      \InputConditions{A triangle \(\ijk\) with a rotation $\fieldRotationAngle_{\ij}^0 \in [-\pi, \pi)$ per edge $\ij$, the triangle Gaussian curvature $\Omega_{\ijk}^0 \in [-\pi, \pi)$ and field magnitude $|z_i| \in \mathbb{R}_{>0}$ per vertex $i$. The triangle must be singular of index $\pm 1$, \ie{} $\fieldRotationAngle_{\ij}^0 + \fieldRotationAngle_{\jk}^0 + \fieldRotationAngle_{\ki}^0 + \Omega_{\ijk}^0 = \pm 2\pi$}
      \OutputConditions{The barycentric coordinates $(b_i,b_j,b_k)$ of the zero.}
      \For {$t = 0 \ldots 1$} \Comment{Interpolate from flat to curved triangle}
         \State $\fieldRotationAngle_{\ij} \gets \fieldRotationAngle_{\ij}^0 + (1-t) \tfrac{1}{3} \left( \Omega_{\ijk}^0 - 2\fieldRotationAngle_{\ij}^0 + \fieldRotationAngle_{\jk}^0 + \fieldRotationAngle_{\ki}^0 \right)$
        \State $\fieldRotationAngle_{\jk} \gets \fieldRotationAngle_{\jk}^0 + (1-t) \tfrac{1}{3} \left( \Omega_{\ijk}^0 + \fieldRotationAngle_{\ij}^0 - 2\fieldRotationAngle_{\jk}^0 + \fieldRotationAngle_{\ki}^0 \right)$
        \State $\fieldRotationAngle_{\ki} \gets \fieldRotationAngle_{\ki}^0 + (1-t) \tfrac{1}{3} \left( \Omega_{\ijk}^0 + \fieldRotationAngle_{\ij}^0 + \fieldRotationAngle_{\jk}^0 - 2\fieldRotationAngle_{\ki}^0 \right)$
        \State $\Omega_{\ijk} \gets t \Omega_{\ijk}^0$
        \State Find $b_j,b_k$ solution of Equation~\ref{eq:zero_system} using Newton method.
      \EndFor
      \State \Return $(1-b_j-b_k,b_j,b_k)$
   \end{algorithmic}
\end{algo}

\begin{algo}{\Proc{MapVertex}$(A, B, \connectionUnitary{}^B, \Omega^B, Z, v_A)$}
   \label{alg:MapVertex}
   \begin{algorithmic}[1]
      \InputConditions{Triangle meshes \(A = (V_A,E_A,F_A)\) and \(B=(V_B,E_B,F_B)\), with the connection $\connectionUnitary \in \CC^{E_B}$ on \(B\), curvature $\Omega\in\RR^{F_B}$ on \(B\), a section \(Z \in \RR^{V_A \times V_B}\) encoding a map from \(A \to B\), and a vertex \(v_A \in V_A\) which we would like to map.}
      \OutputConditions{The image of \(v_A\) on \(B\), given as a face \(\ijk \in F_B\) and barycentric coordinates \((b_i, b_j, b_k)\) recording specific point in face \(\ijk\) that \(v_A\) is mapped to.}

      \State \(z^{(v)} \gets Z_{v_A, \bullet}\) \Comment{Take \(v_A\)'th row of \(Z\) as a section on \(B\)}
      \State \(\fieldRotationAngle_{\ij} \gets \arg\left(\frac{z^{(v)}_j}{\connectionUnitary^B_{\ij} z^B_i}\right)\) \Comment{\cref{eq:AngularOneForm} for each \(\ij \in E_B\)}
      \State \(\text{ind}^{z}_{\ijk} \gets \tfrac{1}{2\pi}\left(\dsf_{1}\fieldRotationAngle + \Omega^B\right)_{\ijk}\) \Comment{\cref{eq:IndexForm} for each \(\ijk \in F_B\)}
      \State \(\ijk \gets \text{first face with}\;\text{ind}^{z}_{\ijk} \neq 0\)
      \State \((b_i, b_j, b_k) \gets \Proc{FindTriangleZero}(\ijk, \fieldRotationAngle, \Omega^B, |z^{(v)}|)\)

      \State \Return $\ijk, (b_i, b_j, b_k)$
   \end{algorithmic}
\end{algo}

\newpage

\begin{algo}{\Proc{ComputeCorrespondence}$(A,B,\varphi,\psi,${\tiny$\{(l^A_k,l^B_k)\}, \{(\gamma^A_k, \gamma^B_k)\}$}$)$}
    \label{alg:ComputeCorrespondence}
    \begin{algorithmic}[1]
        \InputConditions{Triangle meshes \(A = (V_A,E_A,F_A)\) and \(B=(V_B,E_B,F_B)\), and optionally: vertex to face maps \(\varphi : V_A \to F_B, \psi : V_B \to F_A\), pairs of matching landmark points \(\{(l^A_k, l^B_k)\}\) on \(A\) and \(B\) respectively, and/or pairs of matching landmark curves \(\{\gamma^A_k, \gamma^B_k)\}\) on \(A\) and \(B\) respectively.}
        \OutputConditions{A discrete section $z:V_{A\times B} \to \CC$ encoding our optimal bijection between $A$ and $B$}
        \CommentLine{Normalize inputs}
        \State \(A \gets \Proc{NormalizeToUnitSurfaceArea}(A)\)
        \State \(B \gets \Proc{NormalizeToUnitSurfaceArea}(B)\)
        \CommentLine{Build connections on \(A\) and \(B\) (\cref{sec:ConstructingSurfaceConnections})}
        \State \(r^A, \Omega^A \gets \Proc{SurfaceConnection}(A)\)\Comment{\cref{alg:SurfaceConnection}}
        \State \(r^B, \Omega^B \gets \Proc{SurfaceConnection}(B)\)
        \CommentLine{Build FEM Matrices (\cref{sec:FiniteElementSpace}; \cref{alg:BuildFEMConnectionMatrices})}
        \State \(L^\nabla_A, M^\nabla_A \gets \Proc{BuildFEMConnectionMatrices}(A, r^A, \Omega^A)\)
        \State \(L^\nabla_B, M^\nabla_B \gets \Proc{BuildFEMConnectionMatrices}(B, r^B, \Omega^B)\)

        \CommentLine{Build pinning potential (\cref{sec:Landmarks})}
        \State \(V \gets [1] \in \RR^{V_A \times V_B}\)
        \State {\(\sigma_A \gets 1, \sigma_B \gets 1\)}
        \Comment{Landmark penalty (\cref{sec:Landmarks})}
        \If {landmark points were provided}
            \For {\(i_A \in V_A, i_B \in V_B\)}
            \State \(V_{i_A, i_B} \gets 1 - \max_k \exp\big( -\tfrac{1}{2\sigma_A^2}d_A(i_A, l_k^A)^2 - \tfrac{1}{2\sigma_{B}^2}d_{B}(i_B, l_k^B)^2 \big)\)
            \EndFor
        \EndIf
        \If {landmark curves were provided}
            \For {\(i_A \in V_A, i_B \in V_B\)}
            \State \(V^\textsf{C} \gets 1 - \max_k \exp\big( -\tfrac{1}{2\sigma_A^2}d_A(i_A, \gamma_k^A)^2 - \tfrac{1}{2\sigma_{B}^2}d_{B}(i_B, \gamma_k^B)^2 \big)\)
            \State \(V_{i_A, i_B} \gets \min(V_{i_A, i_B}, V^C)\)
            \EndFor
        \EndIf
        \CommentLine{Find initial section \(z_0\)}
        \If {\(\varphi\) and \(\psi\) were provided}
        \State \(z_0 \gets \Proc{MapInitialize}(A, B, r^A, r^B, \Omega^A, \Omega^B, \varphi, \psi)\) \Comment{\cref{alg:MapInitialize}}
        \Else
        \State \(z_0 \gets \Proc{RandomComplexMatrix}(|V_A|, |V_B|)\)
        \EndIf
        \CommentLine{Set Ginzburg-Landau parameter based on eigenvalues of \(A\) and \(B\) (\cref{sec:DiscreteGinzburgLandauMinimization})}
        \State \(\lambda \gets 100 (\Proc{MinEigenvalue}(L^\nabla_A, M^\nabla_A) + \Proc{MinEigenvalue}(L^\nabla_B, M^\nabla_B))\)
        \CommentLine{Optimize Ginzburg-Landau energy (\cref{alg:GinzburgLandau})}
        \State \(z \gets \Proc{LBFGS}(z \mapsto \Proc{GinzburgLandau}(A, L^\nabla_A, M^\nabla_A, B, L^\nabla_B, M^\nabla_B, z, \lambda, V), z_0)\)
        \State \Return $z$
    \end{algorithmic}
\end{algo}

\newpage

\begin{algo}{\Proc{MapInitialize}$(A,B,r^A,r^B,\Omega^A,\Omega^B,\varphi,\psi)$}
    \label{alg:MapInitialize}
    \begin{algorithmic}[1]
        \InputConditions{Triangle meshes \(A = (V_A,E_A,F_A)\) and \(B=(V_B,E_B,F_B)\) with edge lengths \(\ell_{A}, \ell_B\), the surface mesh connections $r^A, r^B$ with curvatures $\Omega^A\in\Omega^2(A)$ and $\Omega^B\in \Omega^2(B)$, and vertex to face maps $\varphi:V_A\to F_B$, $\psi:V_B\to F_A$}
        \OutputConditions{A discrete section $z:V_{A\times B} \to \CC$ approximating the graphs of $\varphi$ and $\psi$}
        \CommentLine{First, build a connection \(r^{\varphi, \psi}\) on the product space}
        \State $L_A \gets \mathsf{d}^A_1 (*^A_1)^{-1} (\mathsf{d}^{A}_{1})^\top$ \Comment{2-form Laplacian on $A$}
        \State $L_B \gets \mathsf{d}^B_1 (*^B_1)^{-1} (\mathsf{d}^{B}_{1})^\top$\Comment{2-form Laplacian on $B$}
        \For {$v_A\in V_A$} \Comment{Rig connection on \(B\)-slices to map \(v_A\) to \(\varphi(v_A)\)}
            \State \(\widetilde\Omega_{v_A, \bullet} \gets 2\pi \delta_{\varphi(v_A)}\) \Comment{target slice curvature}
            \State \(u\gets\Proc{LinearSolve}(L_B, \Omega^B - \widetilde \Omega_{v_A, \bullet})\)
            \State $\rho \gets (*^{B}_{1})^{-1} (\mathsf{d}^{B}_{1})^\top u$
            \State $r^{\varphi,\psi}_{(v_A,\ij)} \gets \exp(\imath\,\rho_{\ij})\, r^B_{\ij}$ \Comment{for every edge $\ij\in E_B$}
        \EndFor
        \For {$v_B\in V_B$} \Comment{Rig connection on \(A\)-slices to map \(v_B\) to \(\psi(v_B)\)}
            \State \(\widetilde \Omega_{\bullet, v_B} \gets 2\pi \delta_{\psi(v_B)}\) \Comment{target slice curvature}
            \State \(u\gets\Proc{LinearSolve}(L_A, \Omega^A - \widetilde \Omega_{\bullet, v_B})\)
            \State $\rho \gets (*^{A}_{1})^{-1} (\mathsf{d}^{A}_{1})^\top u$
            \State $r^{\varphi,\psi}_{(v_B,\ij)} \gets \exp(\imath\,\rho_{\ij})\, r^A_{\ij}$ \Comment{for every edge $\ij\in E_A$}
        \EndFor
        \State \Return \(\begin{aligned}[t]
        &\Proc{MinEigenvector}(\\[-2mm]&\hspace{-5mm}Z \mapsto \Proc{SlicewiseConnectionLaplacian}(A, B, r^{\varphi,\psi}, \widetilde\Omega, Z),\\[-2mm]&\hspace{-5mm}Z \mapsto \Proc{ApplyMassMatrix}(A, B, Z))\\[-3mm]\phantom{}\end{aligned}\)
    \end{algorithmic}
\end{algo}

\begin{algo}{\Proc{GinzburgLandau}$(A, L^\nabla_A, M^\nabla_A, B, L^\nabla_B, M^\nabla_B, z, \lambda, V)$}
    \label{alg:GinzburgLandau}
    \begin{algorithmic}[1]
        \InputConditions{Triangle meshes \(A = (V_A,E_A,F_A)\) and \(B=(V_B,E_B,F_B)\) with connection Laplacians \(\smash{L^\nabla_A \in \CC^{V_A \times V_A}, L^\nabla_B\in\CC^{V_B \times V_B}}\), scalar mass matrices \(\smash{M_A \in \RR^{V_A \times V_A}, M_B \in \RR^{V_B \times V_B}}\), as well as

        the current section \(z \in \CC^{V_A \times V_B}\), the Ginzburg-Landau parameter \(\lambda \in \RR_{>0}\), and the pinning potential \(V \in \RR^{V_A \times V_B}\).}
        \OutputConditions{The Ginzburg-Landau energy \(\DiscreteGinzburgLandau_\lambda\) and its gradient \(\nabla_z\DiscreteGinzburgLandau_\lambda\)}
        \State \(U \gets 0 \in \RR^{V_A \times V_B}\)
        \For {\(i \in V_A, j \in V_B\)}
            \State \(U_{i,j} \gets |z_{i,j}|^2-V_{i,j}\)
        \EndFor
        \State \(\begin{aligned}[t]&\DiscreteGinzburgLandau_\lambda \gets \frac{1}{2}\operatorname{Re} \tr\left[z^\dagger \left(L^\nabla_A z \left(M^\nabla B\right)^\top + M^\nabla_A z \left(L^\nabla_B\right)^\top\right)\right] \hspace{6mm}\text{\Comment{Eq.\ref{eq:DiscreteGL}}}\\
        &\hspace{8mm}+ \frac{\lambda}{4}\tr\left[U^\top M_A U M_B^\top\right]\hspace{1mm}\text{\Comment{using real mass matrices \(M_A, M_B\)}}\end{aligned}\)
        \State \(\begin{aligned}[t]\nabla_z\DiscreteGinzburgLandau_\lambda \gets &L_A^\nabla z \left(M^\nabla_B\right)^\top + M_A^\nabla z \left(L^\nabla_B\right)^\top \\
         &+ \lambda \left(M_A U M_B^\top\right) \odot z\hspace{2mm}\text{\Comment{element-wise product}}\end{aligned}\)\Comment{Eq.\ref{eq:DiscreteGLGrad}}
        \State \Return $\DiscreteGinzburgLandau_\lambda, \nabla_z \DiscreteGinzburgLandau_\lambda$
    \end{algorithmic}
\end{algo}

\newpage

\begin{algo}{\Proc{BuildFEMConnectionMatrices}$(A, r, \Omega)$}
    \label{alg:BuildFEMConnectionMatrices}
    \begin{algorithmic}[1]
        \InputConditions{A triangle mesh \(S = (V,E,F)\) with connection $\connectionUnitary \in \CC^{E}$ and curvature $\Omega\in\RR^{F}$.}
        \OutputConditions{The connection Laplacian \(L^\nabla\)}
        \State \(L^\nabla, M^\nabla \gets 0 \in \CC^{V \times V}, 0 \in \RR^{V \times V}\)
        \For {\textbf{corner} \(\corner{i}{jk} \in S\)}
            \Statex\hspace{\algorithmicindent}\lComment{Laplace matrix}
            \State \(\alpha \gets (\ell_{\ij}^2 - \ell_{\jk}^2 + \ell_{\ki}^2)/2\)\Comment{\(\;\alpha = \langle p_j-p_i, p_k-p_i\rangle\)}
            \State \(w_L \gets \frac{\overline{\connectionUnitary}_{\jk}}{a_{\ijk}} * \left[\left(\ell_{\ij}^2 + \ell_{\ki}^2\right) * f_1(\Omega_{\ijk}) + \alpha\, f_2(\Omega_{\ijk})\right]\)\Comment{Algs. \ref{alg:f1},\ref{alg:f2}}
            \State \(L^\nabla_{\jk} += w_L, \quad L^\nabla_{kj} += \overline{w_L}, \quad L^\nabla_{ii} += \frac{1}{4a_{\ijk}}\left(\ell_{\jk}^2 + \Omega^2_{\ijk}\frac{\ell_{\ij}^2 + \alpha + \ell_{\ki}^2}{90}\right)\)
            \Statex\hspace{\algorithmicindent}\lComment{Mass matrix}
            \State \(w_M \gets a_{\ijk}\, \overline{r}_{\jk}\,f_0(\Omega_{\ijk})\)\Comment{Alg. \ref{alg:f0}}
            \State \(M^\nabla_{\jk} += w_M, \quad M^\nabla_{\kj} += \overline{w_M}, \quad M^\nabla_{ii} += \frac{1}{6} a_{\ijk}\)
        \EndFor
        \State \Return \(L^\nabla, M^\nabla\)
    \end{algorithmic}
\end{algo}

\begin{algo}{\(f_0(s)\)\Comment{Helper for \Proc{BuildFEMConnectionMatrices}}}
    \label{alg:f0}
    \begin{algorithmic}[1]
        \InputConditions{A real number \(s \in \RR\).}
        \OutputConditions{The value \(f_0(s)\) used in \citet[Eq.17]{knoppel2013globally}.}
        \CommentLine{The function \(f_0(s)\) has a removable singularity at \(s=0\). One can use the Chebyshev expansion provided by Kn\"{o}ppel \etal, or the simple Taylor expansion given here}
        \If {\(|s| < \frac{1}{10}\)}
            \State \Return \( \frac{1}{12} - \frac{s^2}{360} + \frac{s^4}{20160} + \imath \left(\frac{s}{60} - \frac{s^3}{2520} + \frac{s^5}{181440}\right)\)
        \Else
            \State \Return \(\frac{1}{3s^4}\left(-6 - 6\imath s + 3s^2 \imath s^3 + 6e^{\imath s}\right)\)
        \EndIf
    \end{algorithmic}
\end{algo}

\begin{algo}{\(f_1(s)\)\Comment{Helper for \Proc{BuildFEMConnectionMatrices}}}
    \label{alg:f1}
    \begin{algorithmic}[1]
        \InputConditions{A real number \(s \in \RR\).}
        \OutputConditions{The value \(f_1(s)\) defined by \citet[\S6.1.1]{knoppel2013globally}.}
        \If {\(|s| < \frac{1}{10}\)}\Comment{Taylor expansion to handle singularity}
            \State \Return \(\frac{s^2}{120} -\frac{s^4}{2688} + \frac{s^6}{129600} + \imath \left(-\frac{s}{24} + \frac{s^3}{504} - \frac{s^5}{17280}\right) \)
        \Else
            \State \Return \(\frac{1}{s^4}\left(3 + \imath s + \frac{s^4}{24} - \imath \frac{s^5}{60} + (-3 + 2 \imath s + \frac{s^2}{2}) e^{\imath s}\right)\)
        \EndIf
    \end{algorithmic}
\end{algo}

\begin{algo}{\(f_2(s)\)\Comment{Helper for \Proc{BuildFEMConnectionMatrices}}}
    \label{alg:f2}
    \begin{algorithmic}[1]
        \InputConditions{A real number \(s \in \RR\).}
        \OutputConditions{The value \(f_2(s)\) defined by \citet[\S6.1.1]{knoppel2013globally}.}

        \If {\(|s| < \frac{1}{10}\)}\Comment{Taylor expansion to handle singularity}
            \State \Return \( -\frac{1}{4} + \frac{s^2}{45} - \frac{s^4}{1120} + \frac{s^6}{56700} + \imath \left( -\frac{s}{24}+ \frac{5s^3}{1008} - \frac{7s^5}{51840}\right)\)
        \Else
            \State \Return \(\frac{1}{s^4}\left(4 + \imath s - \imath\frac{s^3}{6} - \frac{s^4}{12}  + \imath \frac{s^5}{30} + (-4 + 3 \imath s + s^2) e^{\imath s}\right)\)
        \EndIf
    \end{algorithmic}
\end{algo}

\newpage
\begin{algo}{\Proc{SlicewiseConnectionLaplacian}$(A, B, r, \Omega, Z)$}
    \label{alg:SlicewiseConnectionLaplacian}
    \begin{algorithmic}[1]
        \InputConditions{Triangle meshes \(A = (V_A,E_A,F_A)\), \(B=(V_B,E_B,F_B)\), a product space connection \(\connectionUnitary \in \CC^{E_{A \times B}}\) with curvature \(\Omega \in \RR^{F_{A \times B}}\), and a section \(Z \in \CC^{V_A \times V_B}\)}
        \OutputConditions{Applies the product space connection Laplacian (\cref{eq:SlicewiseLaplacian}) for connection \(\connectionUnitary\) to the section \(Z\) to obtain a new section \(Z'\).}
        \For {\(v \in V_A, w \in V_B\)}
            \State \(L^{\nabla,v}_B, \_ \gets \Proc{BuildFEMConnectionMatrices}(B, \connectionUnitary_{v, \bullet}, \Omega_{v, \bullet})\)
            \State \(L^{\nabla,w}_A, \_ \gets \Proc{BuildFEMConnectionMatrices}(A, \connectionUnitary_{\bullet, w}, \Omega_{\bullet, w})\)
            \State \(Z'_{v,w} \gets (M_A)_{v,v} \left(L^{\nabla,v}_B Z^\top\right)_{w,v} + (M_B)_{w,w} \left(L^{\nabla,w}_A Z\right)_{v,w}\)
        \EndFor
        \State \Return \(Z'\)
    \end{algorithmic}
\end{algo}

\begin{algo}{\Proc{ApplyMassMatrix}$(A, B, Z)$}
    \label{alg:ApplyMassMatrix}
    \begin{algorithmic}[1]
        \InputConditions{Triangle meshes \(A = (V_A,E_A,F_A)\), \(B=(V_B,E_B,F_B)\), and a section \(Z \in \CC^{V_A \times V_B}\)}
        \OutputConditions{Applies the product space mass matrix to the section \(Z\).}
        \State \Return \(M_A Z' M_B^\top\)
    \end{algorithmic}
\end{algo}

~\newpage
\section{Parameters and Meshes Statistics}
\label{app:Statistics}
The parameters used for each figure are reported in Table~\ref{tab:statistics}.

\begin{table}[hb]
\caption{Parameters for each figure: number of variables, Ginzburg-Landau parameter schedule, pinning parameter (if applicable) and usage of intrinsic triangulation.}
\label{tab:statistics}
\begin{center}
\rowcolors{2}{aliceblue}{white}
\begin{tabular}{llcccc}
& $|V_A| \times |V_B|$ & $\lambda = t\lambda_0$ & $\sigma$ & iDT  \\
\hline
Fig.~\ref{fig:Teaser} 				        & $4643\times4818$  & $10,50,100$ & $\times$ & \checkmark \\
Fig.~\ref{fig:FmapsUntangling} 			& $1006 \times 998$ &  $10$, $100$ & $\times$ &  \checkmark \\
Fig.~\ref{fig:SymmetryMap}	\figloc{left}	& $3001 \times 3001$ & $10$, $100$ & $\times$ & $\times$ \\
Fig.~\ref{fig:SymmetryMap} \figloc{right}	& $2999 \times 2999$ & $10$, $100$ & $\times$ & $\times$ \\
Fig.~\ref{fig:ZeroPlacement} 				& $502\times2582$  & $100$ & $\times$ & $\times$  \\
Fig.~\ref{fig:MultiresolutionSection} 		& $5000\times 4871$ & $100$ & $\times$ & \checkmark \\
Fig.~\ref{fig:Landmarks} 					& $2485\times 2429$ & $100$ & $1$ & \checkmark  \\
Fig.~\ref{fig:CurvePinning} 				& $3138\times3146$ & $75$ & $\sfrac{1}{10}$ & $\times$  \\
Fig.~\ref{fig:Boundary} 					& $617\times1338$ & $100$ & $1$ & \checkmark  \\
Fig.~\ref{fig:MeshIndependence} \figloc{left} 	& $512\times 1000$  & $100$ & $\times$ &  \\
Fig.~\ref{fig:MeshIndependence} \figloc{right} 	& $4098 \times 2500$  & $100$ & $\times$ &  \\
Fig.~\ref{fig:NnInit} \figloc{top}					& $3000 \times 3000$ & $100$ & $\times$ &  $\times$ \\
Fig.~\ref{fig:NnInit} \figloc{bottom}				& $2000 \times 2000$ & $100$ & $\times$ &  $\times$ \\
Fig.~\ref{fig:LowResolution} 				& $474\times 512$ & $50$ & 1 & $\times$ \\
Fig.~\ref{fig:BiologicalData} \figloc{left}   & $1502\times 1477$ & $100$, $50$ & $\times$ & \checkmark \\
Fig.~\ref{fig:BiologicalData} \figloc{center} & $1502 \times 1477$ & $100$, $50$ & $\times$ & \checkmark \\
Fig.~\ref{fig:BiologicalData} \figloc{right}  & $3000 \times 3000$ & $100$, $50$ & $\times$ & \checkmark  \\
Fig.~\ref{fig:PointCurveConstraints} 		& $954 \times 921$ & $100$ & $\sfrac{1}{\sqrt{20}}$ &  $\times$ \\
Fig.~\ref{fig:FmapsCollapse} 				& $2252 \times 2277$ & $100$ & $\times$ & \checkmark \\
Fig.~\ref{fig:ComparisonInitialization} 	& $3222 \times 6121$ & $10,100$ & $\times$ &  $\times$ \\
Fig.~\ref{fig:ComparisonInterSurfaceMaps} & $4593\times 4017$ & $100$ & $\times$ & $\times$  \\
Fig.~\ref{fig:FemursComparison} \figloc{top}  & $1500\times 1500$ & $100$ & $\times$ & $\times$  \\
Fig.~\ref{fig:FemursComparison} \figloc{middle}& $1500\times 1500$ & $100$ & $\times$ & $\times$  \\
Fig.~\ref{fig:FemursComparison} \figloc{bottom} & $1500\times 1500$ & $100$ & $\times$ & $\times$  \\
Fig.~\ref{fig:OrbifoldComparison} \figloc{left}  & $2000\times 2019$ & $100$ & $\sfrac{1}{4}$ & $\times$  \\
Fig.~\ref{fig:OrbifoldComparison} \figloc{middle}& $2000\times 3000$ & $100$ & $\sfrac{1}{4}$ & $\times$  \\
Fig.~\ref{fig:OrbifoldComparison} \figloc{right} & $2000\times 3000$ & $100$ & $\sfrac{1}{4}$ & $\times$  \\
Fig.~\ref{fig:NonIsometric} 				& $3863\times 3863$ & $25,75$ & $\times$ & \checkmark   \\
\hline
\end{tabular}
\end{center}
\end{table}

\clearpage

\end{document}